\documentclass[12pt,letterpaper]{article}
\usepackage[doublespacing]{setspace}
\usepackage{amssymb,amsmath,amsthm,amsfonts,geometry,graphicx,caption,natbib,booktabs,xcolor,bm,float,array}
\usepackage{lmodern}
\PassOptionsToPackage{hyphens}{url}
\usepackage[colorlinks=true,linkcolor=cyan,urlcolor=cyan,citecolor=cyan,hypertexnames=false]{hyperref}
\newtheorem{thm}{Theorem}
\newtheorem{lem}{Lemma}
\newtheorem{prop}{Proposition}
\theoremstyle{definition}
\newtheorem{assumption}{Assumption}
\newtheorem{defn}{Definition}

\newif\ifanon
\anonfalse

\newcommand{\E}{\mathbb{E}}
\newcommand{\cov}{\operatorname{cov}}
\newcommand{\var}{\operatorname{var}}

\title{Identification, estimation and inference in panel vector autoregressions using external instruments}
\author{Raimondo Pala\thanks{Email: raimondopala@gmail.com. Webpage: \url{https://raimondo-p.github.io}.}\\University of Rome Tor Vergata}
\date{\today}

\begin{document}
\maketitle

\begin{abstract}
\noindent This paper proposes an identification inspired by the SVAR-IV literature that uses external instruments to identify panel vector autoregressions (PVARs), and discusses associated issues of identification, estimation, and inference. I introduce a form of local average treatment effect, the $\mu$-LATE, which arises when a continuous instrument targets a policy variable with linear effects, and which reduces to the local average treatment effect when the policy and the instrument are binary. Under standard assumptions of independence, exclusion, and monotonicity, I show that externally instrumented PVARs estimate the $\mu$-LATE. I construct confidence sets for the impulse responses based on the Anderson-Rubin statistic and, with several instruments, on conditional likelihood ratio tests, which are robust to weak instruments and to common shocks across units. I apply the method to the effect of stock market wealth on local labor markets across US counties.

\end{abstract}
\clearpage

\section{Introduction}

Panel vector autoregressions (PVAR) are one of the standard tools for the estimation of dynamic causal effects in macroeconomics. However, their causal interpretation is frequently limited. Sometimes they are described as identifying a temporally related causality, ``Granger causality''; sometimes, as discussed in \citet{pala2024pvarcontrol}, they may possess a contemporaneous causal interpretation under either exogeneity assumptions, or a suitable control type of assumption. However, Granger causality is not contemporaneous causality\footnote{In the words of \citet{granger2014forecasting}: ``A better term might be temporally related, but since it is such a simple term we shall continue to use it.''}; exogeneity is hard to satisfy due to the rich endogeneity among macroeconomic variables \citep{Nakamura2018a}; and the control type of approach proposed by \citet{pala2024pvarcontrol} may not be satisfied if some units cannot act as controls for some others. This paper shows that when all the previous cases fail, external instruments could be utilized to carry causal claims.

By implementing the same approach as \citet{gertler2015monetary}; \citet{Mertens2013,Mertens2014}; \citet{StockWatson2018}; \citet{Olea2021}, known as proxy-SVAR or SVAR-IV, in PVARs it is possible to retrieve a causal estimand that I define as $\mu$-LATE. The use of instruments for causal inference dates back to \citet{imbens1994identification}, and is fairly common in panel data. A large applied literature makes use of shift-share instruments, which interact an aggregate shift with unit-level exposures, to estimate the effects of government spending, wealth, trade, and credit on regional economies \citep{bartik1991benefits,Nakamura2014,goldsmith2020bartik,chodorowreich2021}. However, many studies generally prefer running static panel regressions, rather than identifying a dynamic multiplier.

Using instruments in PVARs opens up the avenue for several discussions. The first problem is what the instrument identifies when the effect of the policy differs across units, as it typically does across regions or firms. The second is how the dynamic panel bias of the estimated autoregressive coefficients \citep{nickell1981biases} affects the responses at longer horizons, where it accumulates. The third is how to conduct inference when the aggregate shift of the instrument makes all units move together in the same period and when the instrument may be weak.

First, I show that PVARs identified with external instruments identify a weighted average of unit-level effects of a one-unit change in the policy, with two biases. The first one depends on the correlation between the instrument and the untreated outcomes. The second one depends on the correlation between the instrument and the response itself. Framing the issue as a potential outcome framework, and assuming that the instrument is independently extracted with respect to the potential outcome of the outcome variable, makes the bias vanish. In this context, the $\mu$-LATE is the continuous-instrument analogue of the Average Causal Response in \citet{angrist1995twostage}, and reduces to the local average treatment effect in the case of a binary treatment and a binary instrument. I then show that the responses at longer horizons also possess a causal interpretation by assuming that the policy shock is serially uncorrelated and that the instrument is uncorrelated with the other shocks at every date. Finally, I show that using the half-panel jackknife \citep{dhaene2015split} can remove the bias of order $1/T$ that removing the unit means creates in dynamic panels \citep{nickell1981biases}. The confidence sets around the impulse responses are constructed using Anderson-Rubin sets for the case of one instrument, and should be extended to Conditional Likelihood Ratio tests \citep{moreira2003conditional,andrews2006optimal} in the case of multiple instruments. I show that the Anderson-Rubin sets have asymptotically correct coverage independently of instrument strength using $N$ and $T$ asymptotics, and the supplementary material extends the result to the Conditional Likelihood Ratio sets.

The identification result rests on four assumptions on potential outcomes. First, I assume the response needs to be linear on the treatment. Second, instruments need to be independent of the potential outcomes of the outcome variables. Third, I assume the instrument moves the policy of every unit in the same direction, similarly to the monotonicity assumption of \citet{imbens1994identification}. Fourth, to give a causal interpretation to the whole impulse response function, it is necessary that the policy shock is not serially correlated. This assumption is common in time-series econometrics \citep{StockWatson2018}.

Through Monte Carlo simulations I illustrate four points. First, confidence sets clustered by unit alone cover the true response far less often than their nominal level, because the common shocks of the design make the units dependent within a period. Second, if the economist does not use the jackknife correction, the coverage of the sets declines with the horizon. Third, when the instrument is weak and the policy is endogenous, the Anderson-Rubin sets keep their coverage while the delta-method intervals do not. Fourth, I show that in the case of multiple instruments, the Conditional Likelihood Ratio tests \citep{moreira2003conditional} have more power than the Anderson-Rubin test.

The application revisits the contribution of \citet{chodorowreich2021}. They assess whether changes in asset prices affect household spending. In this case, the endogeneity issue arises from the ability of stock prices to predict future income \citep{poterba2000stock}. \citet{chodorowreich2021} use the uneven distribution of stock wealth across US counties: when the market rises, counties whose residents hold more stocks receive a larger wealth gain, and local spending shows up in the payroll of the nontradable sector \citep{mian2014explains}. They estimate the response with local projections instrumented by the interaction of lagged county wealth with the aggregate market return. I revisit their result by estimating a PVAR of the stock wealth shock and nontradable payroll growth across US counties, identified with the same instrument. The response of payroll builds within the first two years after a wealth gain and then settles at a permanently higher level, and it is consistent with the local projections of the original study at every horizon. This paper contributes to several streams of the literature. One natural contribution is at the intersection of the causal inference and the time series fields. In this sense, it complements \citet{pala2024pvarcontrol} in the research agenda oriented to give a causal interpretation to panel vector autoregressions and, more broadly, the literature that uses the Rubin causal model to motivate the causal interpretation of time series models \citep{menchettibojinov2022,menchetticipollinimealli2022,RambachanSheppard2021,bojinovshephard2019}. Second, it contributes to the econometric literature by extending the identification of structural VARs through an external instrument to panel VARs. It also shows that, unlike local projections instrumented with the same instrument, the PVAR-IV does not require the instrument to be uncorrelated with future policy shocks \citep{StockWatson2018}, a result shown in Section~S4 of the supplementary material. Third, it extends the common knowledge of the good coverage properties of the Anderson-Rubin and Conditional Likelihood Ratio for the impulse response functions generated by a PVAR-IV, indicating an avenue for the construction of reliable confidence sets. Fourth, it contributes to the applied literature on the wealth effects of the stock market. The paper is organized as follows. Section~\ref{sec:Identification} introduces the causal framework, Section~\ref{sec:Assumptions} the assumptions, and Section~\ref{sec:Theorems} the identification results. Section~\ref{sec:Inference} covers estimation and inference, Section~\ref{sec:MonteCarlo} the Monte Carlo simulations, and Section~\ref{sec:Application} the application. Section~\ref{sec:Conclusions} concludes. The supplementary material contains the proofs, the extension to several instruments, and additional results.

Some notation. For random variables $X$ and $Y$, $\E(X)$ and $\var(X)$ denote the mean and variance of $X$, and $\cov(X,Y)$ their covariance. For a matrix $R$, $R_{a,b}$ is the entry in row $a$ and column $b$, and $e_k$ is the $k$-th column of the identity matrix. $X\perp Y$ denotes independence.

\section{Causal framework\label{sec:Identification}}

Consider a panel composed of units $i=1,\dots,N$, such as different counties in the United States, and periods $t=1,\dots,T$, such as quarters. Let $W_{it}$ be a policy variable, such as the stock wealth shock of county $i$ in quarter $t$. Moreover, let $Y_{j,it}$ for $j=1,\dots,J$ be the outcome variables, such as the growth of payrolls. Consider the PVAR composed by $x_{it}=(W_{it},Y_{1,it},\dots,Y_{J,it})'$, a vector of dimension $m=J+1$, with the policy variable ordered first. The PVAR I will consider is
\begin{equation}
x_{it}=\alpha_i+\Phi x_{i,t-1}+\tilde{x}_{it},\label{eq:main_PVAR}
\end{equation}
where $\alpha_i$ are unit effects, $\Phi$ is an $m\times m$ matrix and $\tilde{x}_{it}$ are the reduced-form innovations. I write one lag to keep the notation compact, so that $p$ lags are handled through the companion form. I will denote $\Omega_{it}=\{\alpha_i,x_{i,t-1}\}$ as the information set, so that $\tilde{x}_{it}$ is the residual of $x_{it}$ after projecting on $\Omega_{it}$. The notation will indicate $Z_{it}$ as the instrument. In the application, I will use the aggregate stock market return multiplied by the past stock wealth of the county, so that counties with more wealth exposed to the market receive a larger push from the same aggregate movement. Each of those variables can be expressed in a potential outcome framework \citep{Rubin1980,imbens1994identification,angrist1995twostage}. I consider $Y_{j,it}(w)$ as the potential outcome $j$ if the policy is extracted to have value $w$, and $W_{it}(z)$ as the value of the policy if the instrument is extracted to have value $z$. The object of interest is the average effect of a one-unit change in the policy over the units that the instrument moves.

\begin{defn}[$\mu$-LATE]\label{def:muLATE}
Let $\theta_{j,it}$ be the effect on outcome $j$ of a one-unit change in the policy for unit $i$ at date $t$, and $\pi_{it}$ the response of the policy to a one-unit change in the instrument. The $\mu$-LATE of the policy on outcome $j$ is
\[
\mu\text{-LATE}_j=\frac{\E[\theta_{j,it}\,\pi_{it}]}{\E[\pi_{it}]}.
\]
\end{defn}

The $\mu$-LATE weights the effect of each unit by how much the instrument moves its policy. Counties whose wealth responds strongly to the market return count more, and counties whose residents hold no stocks do not count at all. With constant effects, $\theta_{j,it}=\theta_j$, the $\mu$-LATE equals the structural coefficient $\theta_j$, and with heterogeneous effects it is an average effect of a one-unit change in the policy, with weights given by the first-stage responses themselves. Finally, notice that the expectations are taken over units and over time. The next lemma characterizes the probability limit of the 2SLS estimator run on the residuals of the PVAR on the instrument.

\begin{lem}[Frisch-Waugh-Lovell]\label{lem:FWL}
Let $\tilde{W}_{it}$ and $\tilde{Y}_{j,it}$ be the residuals of $W_{it}$ and $Y_{j,it}$ from projections on the regressors in $\Omega_{it}$, and let $\Gamma=\E[\tilde{x}_{it}Z_{it}]$. The coefficient on $W_{it}$ in the two-stage least squares regression of $Y_{j,it}$ on $W_{it}$ with instrument $Z_{it}$ and controls $\Omega_{it}$ equals
\begin{equation}
\beta_j^{IV}=\frac{\sum_{i,t}Z_{it}\tilde{Y}_{j,it}}{\sum_{i,t}Z_{it}\tilde{W}_{it}},\label{eq:betaIV}
\end{equation}
whose probability limit is $\Gamma_j/\Gamma_1$, where $\Gamma_j$ is the entry of $\Gamma$ corresponding to $Y_j$ and $\Gamma_1$ the entry corresponding to $W$.\footnote{Since $\tilde{x}_{it}$ is orthogonal to the regressors in $\Omega_{it}$, residualizing the instrument on them leaves $\Gamma$ and $\beta_j^{IV}$ unchanged.}
\end{lem}

The proof is in Section~S1 of the supplementary material.

\section{Assumptions\label{sec:Assumptions}}

This section is dedicated to the assumptions, which are key to give a causal interpretation to Panel VARs identified using external instruments.

\begin{assumption}[Linearity and exclusion]\label{assu:Linearity}
For each $(i,t)$ there are random coefficients $\theta_{j,it}$ and $\pi_{it}$ such that
\[
Y_{j,it}(w)=Y_{j,it}(0)+\theta_{j,it}\,w,\qquad W_{it}(z)=W_{it}(0)+\pi_{it}\,z,
\]
and the observed variables are $W_{it}=W_{it}(Z_{it})$ and $Y_{j,it}=Y_{j,it}(W_{it})$.
\end{assumption}

Assumption~\ref{assu:Linearity} imposes two conditions: linearity and exclusion. Linearity, in particular, assumes a linear relationship between the outcomes and the treatments, and between the treatments and the exogenous instruments.\footnote{If $W_{it}(z)$ is not linear in $z$, the instrumental variable estimand becomes a density-weighted average of derivatives \citep{angrist1995twostage}. When the policy and the instrument are binary, linearity holds without loss of generality.} For example, in the case of measuring the effect of stock market wealth on local labor markets, this imposes that increasing the wealth gain from 1\% to 2\% of labor income would move the local labor market just as much as moving it from 2\% to 3\%.\footnote{Notice that this imposition rules out the many principal strata that a continuous instrument generates, described in \citet{antonelli2023principal}.}

\begin{assumption}[Independence]\label{assu:Independence}
$(Y_{j,it}(0),W_{it}(0),\theta_{j,it},\pi_{it})\perp Z_{it}\mid\Omega_{it}$ for all $j$, and $\var(Z_{it}\mid\Omega_{it})=\sigma_Z^2>0$.
\end{assumption}

Assumption~\ref{assu:Independence} is a key identification assumption. It assumes that the instrument is as good as randomly assigned with respect to the potential outcomes and the unit-level responses, given the conditioning argument. In essence, this is similar to the conditional exogeneity assumption of \citet{Angrist2011}. Here I also assume that the conditional variance of the instrument is positive, so that the instrument varies given the past, and constant, because otherwise the expectations over units and time would not coincide with the averages of conditional expectations.

\begin{assumption}[Monotonicity]\label{assu:Monotonicity}
Either $\pi_{it}\geq 0$ for all $(i,t)$, or $\pi_{it}\leq 0$ for all $(i,t)$.
\end{assumption}

Assumption~\ref{assu:Monotonicity} requires the instrument to move the policy of every unit in the same direction. It is the continuous counterpart of the no-defiers condition, and it is what makes the weights of the $\mu$-LATE nonnegative. For example, this would impose that a rise in the stock market cannot decrease the stock wealth of any county, since the exposures are nonnegative.

\begin{assumption}[Relevance]\label{assu:Relevance}
$\E[\pi_{it}\mid\Omega_{it}]\neq 0$.
\end{assumption}

Assumption~\ref{assu:Relevance} is required for the estimand's ratio to exist. Violations of it would make the estimand not exist. In practice, the denominator of the estimator in equation~\eqref{eq:betaIV} would converge to zero, and the estimator would not converge to a constant: it would remain a random variable, centered around the ordinary least squares estimand \citep{bound1995problems,StockStaiger1997}. Relevance is not needed for the inference of Section~\ref{sec:Inference}, whose confidence sets remain valid when the instrument is weak.

I now move to a different set of assumptions, concerning the SVAR identification. As in any Panel SVAR, identification in the macroeconomic sense means finding appropriate restrictions and economically driven mechanisms to appropriately bound the response of one variable to another, or to a structural shock. In practice, the economist starts by observing the innovations of the VAR, and assumes that $\tilde{x}_{it}=H\xi_{it}$, where $H$ is an invertible $m\times m$ matrix and the shocks $\xi_{it}$ are mean zero, unit variance, and mutually and serially uncorrelated. One common way of binding the unknown matrix $H$ is by observing the covariance matrix of the residuals, $\Sigma=\E[\tilde{x}_{it}\tilde{x}_{it}']$. In this sense, identification in macroeconomics concerns the fact that the covariance matrix $\Sigma$ provides $m(m+1)/2$ restrictions, fewer than the $m^2$ elements of $H$, so $H$ is not identified from second moments alone. With the same mathematics as \citet{StockWatson2018} and \citet{Olea2021}, it is possible to show that external instruments identify, up to scale, the column of $H$ associated with the shock the economist wants to identify.\footnote{Write $\tilde{x}_{it}=H\xi_{it}=H_{\cdot1}\xi_{1,it}+\sum_{k\geq2}H_{\cdot k}\xi_{k,it}$, where $\xi_{1,it}$ is the policy shock. If the instrument is correlated with the policy shock and uncorrelated with the other shocks, then $\Gamma=\E[\tilde{x}_{it}Z_{it}]=H_{\cdot1}\E[\xi_{1,it}Z_{it}]$, so $\Gamma$ is proportional to the column $H_{\cdot1}$, with an unknown factor of proportionality $\E[\xi_{1,it}Z_{it}]$. Dividing by the first entry removes it: $b=\Gamma/\Gamma_1=H_{\cdot1}/H_{11}$, so that $b_1=1$, and $b_j=\Gamma_j/\Gamma_1=\beta_j^{IV}$ by Lemma~\ref{lem:FWL}. The innovations then become $\tilde{x}_{it}=b\,\eta_{it}+\nu_{it}$, with $\eta_{it}=H_{11}\xi_{1,it}$ and $\nu_{it}=\sum_{k\geq2}H_{\cdot k}\xi_{k,it}$. Because $b_1=1$, a unit value of $\eta_{it}$ moves the policy variable by one unit on impact, so the responses are responses to a one-unit change in the policy, instead of a one-standard-deviation shock, and they keep the interpretation of instrumental variable coefficients.} I use this unit normalization, $b=H_{\cdot1}/H_{11}$, and write $\eta_{it}$ for the policy shock measured in units of the policy variable. In practice, this means that if the economist wants to identify the effects of two policies on one outcome, an instrument that satisfies Assumptions~\ref{assu:Linearity} to~\ref{assu:Relevance} for one of the policies, and is uncorrelated with the other, identifies only the column of the first policy, while an instrument correlated with both identifies a combination of the two columns.

\begin{assumption}[Policy shock and exogeneity]\label{assu:Shock}
The innovations satisfy $\tilde{x}_{it}=b\,\eta_{it}+\nu_{it}$ with $b_1=1$, where $\eta_{it}$ is the policy shock in units of $W$ and $\nu_{it}$ collects the remaining shocks, and
\begin{enumerate}
\item[(i)] $\eta_{it}$ is uncorrelated with $\eta_{i,t+s}$ for $s\neq0$ and with $\nu_{i,t+s}$ for all $s$, and the innovations $\tilde{x}_{it}$ are serially uncorrelated and recoverable from current and past values of $x_{it}$;
\item[(ii)] $\E[Z_{it}\,\nu_{j,t+s}]=0$ for every unit $j$ and every integer $s$, while the instrument may be correlated with policy shocks of any unit at any date;
\item[(iii)] $\E[\ddot{Z}_{it}\,\eta_{it}]\neq0$, where $\ddot{Z}_{it}$ is the instrument after removing the unit and period effects.
\end{enumerate}
\end{assumption}

Part (i) of Assumption~\ref{assu:Shock} is the requirement of the structural VAR literature that the policy shock is a serially uncorrelated innovation that can be recovered from the data \citep{StockWatson2018}. Part (ii) requires the instrument to be uncorrelated with all shocks other than the policy shock, at every date. Lead-lag exogeneity, which \citet{StockWatson2018} use to establish the equivalence between VAR-IV and local projection IV responses, would further require the instrument to be uncorrelated with past and future policy shocks. Part (ii) does not, so the instrument may be persistent: an aggregate return that is correlated with past returns is correlated with past policy shocks, and this is allowed. Part (ii) covers the shocks of every unit because the period effects of the estimator average the shocks of different units, so in the application the market return must be uncorrelated with the shocks other than wealth of every county. Part (iii) is relevance in terms of the policy shock, stated for the variation of the instrument that the estimator uses. With a shift-share instrument and period effects, it requires exposures that differ across units.

\begin{assumption}[Regularity and dependence]\label{assu:Regularity}
Let $X_{it}$ collect the lags of $x_{it}$ and the exogenous controls after removing unit and period effects, $\ddot{Z}_{it}$ the instrument after the same transformation, and $\psi_{it}=\bigl(\mathrm{vec}(X_{it}\tilde{x}_{it}')',\ \ddot{Z}_{it}\tilde{x}_{it}'-\Gamma'\bigr)'$ the moment contributions of the estimator.
\begin{enumerate}
\item[(i)] The eigenvalues of the companion matrix of the PVAR lie inside the unit circle, $\Sigma$ is positive definite, $\E[X_{it}X_{it}']$ is nonsingular, the conditional expectations of $x_{it}$ and $Z_{it}$ given $\Omega_{it}$ are linear in the regressors, and $(x_{it},Z_{it})$ have moments of order $4+\delta$ for some $\delta>0$, bounded uniformly in $(i,t)$.
\item[(ii)] $\cov(\psi_{it},\psi_{js})=0$ whenever $i\neq j$ and $t\neq s$. Within a unit, $\psi_{it}$ is $\alpha$-mixing over $t$, uniformly in $i$, with mixing coefficients $\alpha(k)$ such that $\sum_k\alpha(k)^{\delta/(2+\delta)}<\infty$. Within a period, the contributions of different units may be arbitrarily correlated.
\item[(iii)] $N,T\to\infty$ with $N/T\to\kappa\in(0,\infty)$. The variance of $\sum_{i,t}\psi_{it}$ equals $v_{NT}\,\Omega_{NT}$ for a scalar $v_{NT}$ with $NT\leq v_{NT}\leq C\,N^2T$ for some constant $C$, $\Omega_{NT}\to\Omega$ positive definite, $v_{NT}^{-1/2}\sum_{i,t}\psi_{it}\xrightarrow{d}N(0,\Omega)$, and
\[
v_{NT}^{-1}\Bigl[\sum_i\Bigl(\sum_t\psi_{it}\Bigr)\Bigl(\sum_t\psi_{it}\Bigr)'+\sum_t\Bigl(\sum_i\psi_{it}\Bigr)\Bigl(\sum_i\psi_{it}\Bigr)'-\sum_{i,t}\psi_{it}\psi_{it}'\Bigr]\xrightarrow{p}\Omega.
\]
\end{enumerate}
\end{assumption}

Part (i) of Assumption~\ref{assu:Regularity} requires the system to be stationary, which in the application holds for payroll growth. Part (ii) is the dependence structure under which clustering in two dimensions is valid \citep{cameron2011robust}. It allows the contributions of a county to be correlated over time, and common shocks to move the contributions of different counties in the same quarter. This rules out correlation between different counties at different dates, which holds when the common shocks are serially uncorrelated. Part (iii) discusses two cases. When the units are independent, $v_{NT}$ is of order $NT$ and the estimator converges at the rate $\sqrt{NT}$. When a common shock with a nondegenerate effect enters the contributions, $v_{NT}$ is of order $N^2T$ and the rate is $\sqrt{T}$, because averaging over units does not remove the common component. Its second scope is to state the central limit theorem and the consistency of the two-way clustered sum at the true parameters as high-level conditions. They hold in common cases where the contributions form a separately exchangeable array \citep{davezies2021empirical}, and under the conditions for multiway clustering of \citet{mackinnon2021wild}. Both results impose more structure than part (ii), such as independence of observations that share no unit and no period, and I state the conditions directly so that they cover the dependence of the moments of a PVAR, whose lagged regressors carry past common shocks.

\section{Identification\label{sec:Theorems}}

The first result describes what the instrumental variable estimator of the impact response converges to without the independence and monotonicity assumptions. Let $\mu_{Z,it}=\E[Z_{it}\mid\Omega_{it}]$, and define the weights
\[
\omega_{it}=\frac{\pi_{it}(Z_{it}-\mu_{Z,it})^2}{\E[\pi_{it}(Z_{it}-\mu_{Z,it})^2]}
\]
and the two terms
\begin{gather*}
B_{j}=\E\bigl[(Z_{it}-\mu_{Z,it})\bigl(Y_{j,it}(0)+\theta_{j,it}W_{it}(0)+\theta_{j,it}\pi_{it}\mu_{Z,it}\bigr)\bigr],\\
B_{W}=\E\bigl[(Z_{it}-\mu_{Z,it})\bigl(W_{it}(0)+\pi_{it}\mu_{Z,it}\bigr)\bigr].
\end{gather*}
The term $B_j$ is the covariance of the instrument with the part of the outcome that does not come from the instrument-induced change in the policy, and $B_W$ is the analogous covariance for the policy.

\begin{thm}\label{thm:decomp}
Under Assumption~\ref{assu:Linearity}, if $\E[\pi_{it}(Z_{it}-\mu_{Z,it})^2]\neq0$ and $\Gamma_1\neq0$,
\[
\beta_j^{IV}=\E[\omega_{it}\,\theta_{j,it}]+\Delta_j,\qquad \Delta_j=\frac{B_j-\E[\omega_{it}\theta_{j,it}]\,B_W}{\E[\pi_{it}(Z_{it}-\mu_{Z,it})^2]+B_W},
\]
where the weights $\omega_{it}$ average to one.
\end{thm}

The estimator therefore converges to a weighted average of unit-level effects plus a bias term. The bias $\Delta_j$ is nonzero when the instrument is correlated with what payroll or wealth would have been without the instrument, for instance if the market rose in quarters in which wealthy counties were growing for other reasons. The weights may be negative, which happens when the instrument moves the policy of some units in opposite directions. The causal assumptions remove both problems.

\begin{thm}\label{thm:PVAR-LATE}
Under Assumptions~\ref{assu:Linearity} to~\ref{assu:Relevance}, $B_j=B_W=0$, $\Delta_j=0$, and
\[
\beta_j^{IV}=\frac{\Gamma_j}{\Gamma_1}=\mu\text{-LATE}_j,
\]
with nonnegative weights $\pi_{it}/\E[\pi_{it}]$. The ratio $\beta_j^{IV}$ is the coefficient on $W_{it}$ in a two-stage least squares regression of $Y_{j,it}$ on $W_{it}$, with instrument $Z_{it}$ and controls $\Omega_{it}$.
\end{thm}

The proofs are in Section~S1 of the supplementary material, and neither result requires normality of the innovations or of the instrument. Theorem~\ref{thm:PVAR-LATE} gives the impact response of a PVAR-IV the same causal content as a static instrumental variable regression with the lags as controls. The SVAR-IV literature states the conditions on the instrument as relevance and a zero correlation with the other shocks \citep{StockWatson2018,Olea2021}, which are statistical conditions and do not map the estimator into a causal estimand when effects are heterogeneous. Theorem~\ref{thm:decomp} shows what is lost without the causal assumptions, and Theorem~\ref{thm:PVAR-LATE} shows what they deliver. For practitioners, the difference is the reading of the estimate. Without the result, the impact response is the response of the pooled system to a normalized shock, and its relation to the effect in any given unit is unknown. With the result, it is an average of unit-level effects of a one-unit change in the policy, with weights that are nonnegative and known up to the first-stage responses. The same reading applies to the static panel instrumental variable regressions of the applied literature, whose coefficient is the impact response of a PVAR-IV with no lags. The weights also tell the researcher which units drive the estimate, which matters when the estimate is extrapolated to units or periods with different exposure.

Shift-share instruments give the weights a specific form. Suppose $Z_{it}=s_ig_t$, with a unit-specific exposure $s_i$ and an aggregate shift $g_t$, that the model includes period effects, as it should when the shift is aggregate, and that Assumption~\ref{assu:Independence} holds conditional on the unit, so that only the shifts need to be independent of the coefficients. The period effects remove the average exposure times the shift, so the instrument net of its conditional mean, $Z_{it}-\mu_{Z,it}$, is proportional to $s_i-\bar{s}$, where $\bar{s}$ is the average exposure, its conditional variance is $(s_i-\bar{s})^2\sigma_g^2$, and the argument of Theorem~\ref{thm:PVAR-LATE} gives
\[
\frac{\Gamma_j}{\Gamma_1}=\frac{\sum_i (s_i-\bar{s})^2\,\E[\theta_{j,it}\pi_{it}\mid i]}{\sum_i (s_i-\bar{s})^2\,\E[\pi_{it}\mid i]},
\]
a weighted average of unit-level effects with weights proportional to the squared deviation of the exposure from its average. When the period effects are specific to groups of units, such as states, $\bar{s}$ is the average exposure of the group. Units whose exposure differs most from the average dominate the estimate, which in the application means counties whose stock wealth relative to labor income differs most from that of their state.

Responses beyond impact propagate the impact vector through the dynamics. The next lemma states this, together with the property of the transformed moment that the estimator uses.

\begin{lem}[Impulse response representation]\label{lem:IRF}
Under Assumptions~\ref{assu:Linearity} to~\ref{assu:Regularity} and model~\eqref{eq:main_PVAR}:
\begin{enumerate}
\item[(a)] $\Gamma=b\,\E[Z_{it}\eta_{it}]$, and the moment of the innovations and instrument after removing unit and period effects, $\ddot\Gamma=\E[\ddot{\tilde{x}}_{it}\ddot{Z}_{it}]$, satisfies $\ddot\Gamma=b\,\ddot\Gamma_1$ with $\ddot\Gamma_1=\E[\ddot{\eta}_{it}\ddot{Z}_{it}]\neq0$, so that $b=\ddot\Gamma/\ddot\Gamma_1$ and $b_j$ is the instrumental variable coefficient of outcome $j$ with the transformed instrument, for every $j$;
\item[(b)] the response of $x_{i,t+h}$ to a unit change in $\eta_{it}$, defined as the coefficient on $\eta_{it}$ in the moving average representation of $x_{it}$ after removing the unit effects, is
\[
\text{IRF}_h=\Phi^h b,\qquad h\geq0.
\]
\end{enumerate}
\end{lem}

Combining Theorem~\ref{thm:PVAR-LATE} with part (a), the impact response of each outcome is its $\mu$-LATE. The response at horizon $h\geq1$ propagates the impact vector through the dynamics $\Phi$, so its causal content rests on the potential outcome assumptions and on Assumption~\ref{assu:Shock}(i). The response uses only the contemporaneous covariance $\Gamma$ and the autoregressive dynamics, which is why the instrument may be correlated with policy shocks at other dates.

\section{Estimation and inference\label{sec:Inference}}

In this section, I will discuss three different aspects of estimation. As with any VAR, Panel VARs need to estimate the autoregressive coefficients and provide unbiased estimates with reliable confidence intervals for the regressors. Secondly, because the structural, or causal, identification step proceeds by regressing the residuals of the policy variable against the instrument, that regression needs specific attention. Finally, the process of generating the impulse response functions is a third step.

First, I remove the unit and period effects and estimate $\Phi$ by least squares on the transformed data, the within estimator. For a fixed $T$, the estimator is known to be inconsistent \citep{nickell1981biases,arellano1995another}, because the unit mean of the lagged variables contains the innovations of later periods, with a Nickell bias of order $1/T$ \citep{nickell1981biases}. Longer panels have lower bias, but the bias of each coefficient accumulates in the powers $\Phi^h$, so it grows with the horizon of the response. I therefore correct the estimator with the half-panel jackknife of \citet{dhaene2015split}. I split the observations of each unit into a first and a second half, taking the lags of the first observations of the second half from the end of the first half, estimate $\Phi$ on each half, and compute
\[
\hat\Phi^{J}=2\hat\Phi-\tfrac{1}{2}\bigl(\hat\Phi^{(1)}+\hat\Phi^{(2)}\bigr),
\]
where $\hat\Phi$ is the full-sample estimate and $\hat\Phi^{(1)}$ and $\hat\Phi^{(2)}$ are the half-sample estimates. Each half has half of the observations, so its bias of order $1/T$ is twice that of the full sample, and the combination removes it. Period effects do not add a bias of this kind, because the mean of each period uses only contemporaneous observations, which are uncorrelated with the current innovation of any other unit and enter the regressors of the unit itself only through lagged values.

Second, I compute the residuals $\hat{\tilde{x}}_{it}$ with $\hat\Phi^J$ on the full sample, estimate
\[
\hat\Gamma=\frac{1}{n}\sum_{i,t}\hat{\tilde{x}}_{it}\ddot{Z}_{it},
\]
where $n$ is the number of observations and $\ddot{Z}_{it}$ is the instrument after removing the unit and period effects, and I set $\hat b=\hat\Gamma/\hat\Gamma_1$. This is the key identification step, which produces the causal estimates $\hat b$. This is the two-stage least squares coefficient of each outcome innovation on the policy innovation, with the transformed instrument.\footnote{Notice that removing the unit and period effects changes the scale of $\Gamma$ and leaves its direction unchanged, because the instrument is uncorrelated with the other shocks of every unit at every date (Lemma~\ref{lem:IRF}(a)).} Third, I compute $\widehat{\text{IRF}}_h=(\hat\Phi^J)^h\hat b$ as with any other VAR. When the units carry weights constant over time, such as population weights, all three steps use weighted least squares and weighted moments.\footnote{An alternative removes the unit effects by forward orthogonal deviations and estimates $\Phi$ by GMM with lagged levels as instruments \citep{arellano1995another}. This estimator is consistent for a fixed number of periods, but the jackknife procedure may be preferred, since the GMM instruments are weak when the series are persistent and the unit effects are large relative to the innovations \citep{blundell1998initial}. The supplemental material includes simulations comparing the two estimators.} Two summaries of the responses are useful. When an outcome enters the PVAR in first differences, as payroll does in the application, its level response at horizon $h$ is the cumulated response $\sum_{k=0}^{h}e_j'\text{IRF}_k$. When the policy itself is persistent, as government spending typically is, the literature summarizes the effect by the cumulative multiplier
\[
M_{j,H}=\frac{\sum_{h=0}^{H}e_j'\,\text{IRF}_h}{\sum_{h=0}^{H}e_1'\,\text{IRF}_h},
\]
the cumulative response of outcome $j$ divided by the cumulative response of the policy variable \citep{rameyzubairy2018}.

One known issue in the instrumental variable literature is that of weak instruments. When an instrument is weak, the ratio $\Gamma_j/\Gamma_1$ is unbounded when $\Gamma_1$ approaches zero. The conventional procedure first checks the first-stage $F$-statistic \citep{StockStaiger1997} and then builds standard intervals in the second stage, which do not have correct coverage when the instrument is weak.\footnote{Common statistical packages report bootstrapped confidence intervals for the impulse response functions of Panel VARs. This is the natural choice for the case of standard inference. However, in the case of IV identification resampling the estimated response does not remove the weak-instrument problem, because the response remains a ratio with a denominator that can be close to zero, rendering the bootstrap invalid \citep{jentsch2019dynamic,Olea2021}.} Anderson-Rubin (AR) sets avoid the problem because their null distribution does not depend on the strength of the instrument \citep{Anderson1949,stock2002testing,mikusheva2010robust,AndrewsStockSun2019}, and with one instrument the AR test is optimal \citep{andrews2006optimal}. \citet{Olea2021} build such sets for impulse responses of SVAR-IVs, and I extend their construction to the panel setting and to the CLR statistics. Consider the response of variable $s$ at horizon $h$ normalized by the impact response of the policy, $\lambda_{h,s}=e_s'\Phi^h\Gamma/e_1'\Gamma$. The null $H_0:\lambda_{h,s}=\lambda_0$ is equivalent to $(e_s'\Phi^h-\lambda_0e_1')\Gamma=0$, which is linear in $\Gamma$. Define
\[
\text{AR}_{h,s}(\lambda_0)=\frac{\bigl[(e_s'\hat\Phi^h-\lambda_0e_1')\hat\Gamma\bigr]^2}{\hat\sigma^2_{h,s}(\lambda_0)},\qquad
\text{CS}^{\text{AR}}_{h,s}(1-\alpha)=\{\lambda_0:\ \text{AR}_{h,s}(\lambda_0)\leq\chi^2_{1,1-\alpha}\},
\]
where $\hat\Phi$ is the corrected estimator and $\hat\sigma^2_{h,s}(\lambda_0)$ is the delta-method estimator of the variance of the numerator, based on the joint covariance of $(\mathrm{vec}\,\hat\Phi,\hat\Gamma)$. The same construction gives sets for the other two objects. For the cumulated response, replace $\Phi^h$ by $\sum_{k=0}^h\Phi^k$. For the cumulative multiplier, the null $M_{j,H}=M_0$ is equivalent to $c_H(M_0)'\Gamma=0$ with $c_H(M_0)'=\sum_{h=0}^{H}(e_j-M_0e_1)'\Phi^h$, so $c_H(M_0)$ replaces $(e_s'\Phi^h-\lambda_0e_1')$. The joint covariance is estimated by clustering in the time and unit dimension \citep{cameron2011robust}. Let $\hat V_{\text{unit}}$ be the covariance estimator clustered by unit, which sums the moment contributions of each unit over time, $\hat V_{\text{period}}$ the estimator clustered by period, which sums the contributions of each period over units, and $\hat V_{\text{het}}$ the heteroskedasticity-robust estimator. The two-way estimator is $\hat V=\hat V_{\text{unit}}+\hat V_{\text{period}}-\hat V_{\text{het}}$, with negative eigenvalues set to zero. Clustering by unit accounts for the serial correlation of the contributions of a unit, and clustering by period accounts for the common shocks that move different units in the same period. A special case arises in the case of shift-share instruments \citep{bartik1991benefits,goldsmith2020bartik,borusyak2022quasi}. A shift-share instrument generates the second kind of dependence through its aggregate shift, so clustering by unit alone understates the variance \citep{adao2019shift}. In both cases, reporting the $F$-statistic is always recommended.

\begin{prop}\label{prop:AR-coverage}
Suppose Assumptions~\ref{assu:Linearity} to~\ref{assu:Regularity} hold, with $\Gamma_1\neq0$ possibly drifting to zero as $N,T\to\infty$, and that the variance of the numerator of $\text{AR}_{h,s}(\lambda_0)$, normalized as in Assumption~\ref{assu:Regularity}(iii), is bounded away from zero. Then, with $\hat\Phi=\hat\Phi^{J}$,
\[
\lim_{N,T\to\infty}\Pr\bigl(\lambda_{h,s}\in\text{CS}^{\text{AR}}_{h,s}(1-\alpha)\bigr)=1-\alpha,
\]
and the same holds for the cumulated response and for the cumulative multiplier $M_{j,H}$. With the uncorrected within estimator, the same limits hold if $v_{NT}/(N^2T)$ is bounded away from zero.
\end{prop}

The proof is in Section~S1 of the supplementary material. Notice that instrument strength, or relevance, is not necessary to provide robust confidence sets. The result holds when $\Gamma_1$ is small, in which case the $F$-statistic is small and the set remains valid, possibly unbounded. The proposition also states when the correction is needed asymptotically. With a common shock in the moment contributions, the estimator converges at the rate $\sqrt{T}$, and a bias of order $1/T$ vanishes at that rate. With independent units, the rate is $\sqrt{NT}$, and the bias is of the same order as the sampling error when $N$ and $T$ grow at the same rate, as in \citet{hahnkuersteiner2002}. The correction is valid in both cases, and in finite samples it also matters in the first case, because the bias accumulates over the horizon. When identification is strong, the AR set and the plug-in set of \citet{Olea2021} are asymptotically equivalent. With several instruments, the conditional likelihood ratio test is the natural choice \citep{moreira2003conditional,andrews2006optimal}, and Section~S2 of the supplementary material gives the construction and its size.

\section{Monte Carlo simulations\label{sec:MonteCarlo}}

This section is dedicated to the illustration of simulations calibrated on the empirical application, which will analyze the bias of the within estimator in a panel of the length of the application, weak identification when the policy is endogenous, and inference with multiple instruments. The data are generated for $N$ counties over $T$ quarters, with $N$, $T$ and the lag order $p$ equal to those of the application. The aggregate shift follows
\begin{equation}
g_t=\rho_g\,g_{t-1}+e_t,\qquad e_t\sim N(0,\sigma_e^2),\label{eq:mc_shift}
\end{equation}
where $g_t$ plays the role of the aggregate stock market return and $e_t$ of its unpredictable component. The instrument is the product of the shift and a county exposure,
\begin{equation}
Z_{it}=s_i\,g_t,\label{eq:mc_instrument}
\end{equation}
where $s_i$ corresponds to the lagged stock wealth of county $i$ relative to its labor income. The innovation of the policy variable, the stock wealth shock, is
\begin{equation}
\tilde{W}_{it}=c\,\gamma\,s_i\,e_t+\psi_W\,s_i\,a_t+u_{it},\label{eq:mc_policy}
\end{equation}
so that the policy responds to the instrument only through the innovation $e_t$ of the shift, with loading $\gamma$ scaled by a design constant $c$. The second term is a common shock $a_t\sim N(0,1)$ whose effect grows with the exposure of the county, which captures county returns that depart from the market return in a way related to wealth, and $u_{it}$ is an idiosyncratic shock. The innovation of the outcome, the growth of nontradable payroll, is
\begin{equation}
\tilde{Y}_{it}=\beta\,\tilde{W}_{it}+\psi_Y\,s_i\,f_t+v_{it},\label{eq:mc_outcome}
\end{equation}
where $\beta$ is the causal effect on impact, $f_t\sim N(0,1)$ is a common shock, such as an aggregate demand shock, to which counties with more stock wealth are more exposed, and $v_{it}$ is an idiosyncratic shock correlated with $u_{it}$. The shocks $e_t$, $a_t$ and $f_t$ are mutually independent, and $(u_{it},v_{it})$ are independent over counties and quarters. Finally, the variables follow the PVAR
\begin{equation}
x_{it}=\alpha_i+\sum_{l=1}^{p}\Phi_l\,x_{i,t-l}+\tilde{x}_{it},\qquad \tilde{x}_{it}=(\tilde{W}_{it},\tilde{Y}_{it})',\label{eq:mc_pvar}
\end{equation}
with county effects $\alpha_i$. Notice that the instrument moves the policy only through $e_t$, which is independent of every other shock at every date. This means that this instrument naturally satisfies Assumption~\ref{assu:Shock}(ii). Moreover, $\rho_g$ measures a failure of lead-lag exogeneity, since it indicates correlation with past policy shocks. The object of interest is the cumulated response of the outcome, the response of the level of payroll, which at horizon $h$ equals $\sum_{k=0}^{h}e_2'C_k(1,\beta)'$, where $C_k$ are the moving average matrices implied by $\Phi_1,\dots,\Phi_p$.

The parameters are calibrated to the application, and then estimated without population weights, since the simulated counties have equal weight, and the loading $\psi_W$ is set so that the median two-way clustered first-stage statistic of the calibrated design equals the one in the application. Section~S3 of the supplementary material describes each calibration target.

The first simulation uses the calibrated design, $c=1$. The second weakens the instrument, $c=0.4$ and $c=0.2$. In the calibrated design the policy residual that the instrument does not explain is almost uncorrelated with the outcome residual, and in that case a weak instrument leaves the Wald interval centered on the true value. The robust sets are designed to be robust to weak instruments, whose consequences depend on endogeneity. For this reason, the second simulation makes the outcome innovation correlated with the part of the policy innovation that the instrument does not move, with correlation 0.9, leaving its variance unchanged.\footnote{The last two terms of~\eqref{eq:mc_outcome} are replaced by $\rho\sqrt{v_Y/v_W}\,(\psi_W s_i a_t+u_{it})+\sqrt{1-\rho^2}\,(\psi_Y s_i f_t+v_{it})$ with $\rho=0.9$, where $v_W$ and $v_Y$ are the variances of the two bracketed terms.} The third simulation uses several instruments, to illustrate the differences between the AR and CLR power in the case of multiple IV. To the shift-share instrument~\eqref{eq:mc_instrument} it adds the exposure multiplied by the shift lagged one to three quarters,
\[
Z^{(l)}_{it}=s_i\,g_{t-l},\qquad l=1,2,3,
\]
which researchers often include to strengthen the first stage. These instruments are valid, because they are uncorrelated with every shock other than the policy shock, and they carry no information on the current policy innovation, which responds to the shift only through $e_t$. The design is calibrated, $c=1$, or weak, $c=0.4$, as in the second simulation, and it compares the AR and CLR tests with all four instruments with the AR test that uses the shift-share instrument alone.

Each replication removes the county and quarter effects and estimates the PVAR by the within estimator with $p$ lags, with and without the half-panel jackknife correction. For the cumulated response at horizons $h=0,4,7,12$, I compute AR sets with the covariance clustered by county, by quarter, and in both dimensions, and the delta-method Wald interval with two-way clustering. Each design has 500 replications. Table~\ref{tab:mc} reports the coverage of the nominal 95\% sets clustered by quarter and in both dimensions and the median bias of the point estimate, and Table~S4 of the supplementary material compares the three clusterings.

\begin{table}[ht]\centering\footnotesize\setlength{\tabcolsep}{4pt}
\resizebox{\textwidth}{!}{\begin{tabular}{llr rrr rr rr}\toprule
 & & & \multicolumn{3}{c}{Within estimator, coverage (\%)} & \multicolumn{2}{c}{Jackknife, coverage (\%)} & \multicolumn{2}{c}{Median bias}\\
\cmidrule(lr){4-6}\cmidrule(lr){7-8}\cmidrule(lr){9-10}
Design & $h$ & $F$ & AR, quarter & AR, two-way & Wald, two-way & AR, two-way & Wald, two-way & Within & Jackknife\\\midrule
Calibrated & 0 & 78.1 & 95.2 & 95.4 & 92.2 & 95.4 & 92.2 & 0.01 & 0.01\\
 & 4 &  & 94.0 & 94.6 & 93.4 & 93.8 & 92.4 & -0.04 & 0.08\\
 & 7 &  & 91.6 & 92.4 & 89.2 & 92.0 & 89.4 & -0.14 & 0.09\\
 & 12 &  & 90.0 & 91.2 & 88.2 & 92.2 & 90.4 & -0.20 & 0.05\\
\addlinespace
Weak, $c=0.4$ & 0 & 13.6 & 94.0 & 94.2 & 91.8 & 94.0 & 91.4 & -0.42 & -0.17\\
 & 4 &  & 92.8 & 93.4 & 94.2 & 92.4 & 92.4 & 0.16 & -0.28\\
 & 7 &  & 93.8 & 94.0 & 94.6 & 93.2 & 92.4 & -0.35 & 0.91\\
 & 12 &  & 94.8 & 95.2 & 96.4 & 94.8 & 93.6 & -1.09 & 1.11\\
\addlinespace
Weak, $c=0.2$ & 0 & 3.6 & 92.6 & 93.0 & 85.0 & 93.0 & 84.6 & 2.09 & 2.16\\
 & 4 &  & 94.8 & 95.4 & 90.2 & 94.6 & 89.2 & 1.30 & 2.01\\
 & 7 &  & 94.8 & 95.0 & 90.2 & 94.8 & 89.6 & -0.85 & 1.46\\
 & 12 &  & 94.2 & 94.6 & 90.4 & 94.8 & 90.4 & -0.80 & 2.60\\
\addlinespace
\bottomrule\end{tabular}}
\caption{Coverage of nominal 95\% sets for the response of the outcome at horizon $h$ to a unit change in the policy variable, and median bias of the point estimate. $F$: median first-stage statistic, two-way clustered. AR: Anderson-Rubin set clustered by quarter or two-way. Wald: delta-method interval, two-way clustered. Within: within estimator of the PVAR coefficients. Jackknife: half-panel jackknife correction of the within estimator. Weak designs: correlation 0.9 between the outcome innovation and the part of the policy innovation not moved by the instrument. 500 replications per design.\label{tab:mc}}
\end{table}

In the calibrated design, sets clustered by county alone cover the true response far less often than their nominal level at every horizon (Table~S4 of the supplementary material). The common shocks in~\eqref{eq:mc_policy} and~\eqref{eq:mc_outcome} have effects proportional to the exposure, so the moment contributions of counties are correlated within a quarter, and clustering by county ignores this dependence. Sets clustered by quarter and in both dimensions have close to nominal coverage on impact. Without correction, their coverage declines moderately with the horizon, because the bias of the within estimator in each coefficient of $\Phi_l$ accumulates in the moving average matrices, and the median bias column shows the corresponding underestimation of the response. The jackknife correction reduces the bias and keeps coverage close to the nominal level at every horizon, with sets somewhat wider than without correction.\footnote{Table~S3 of the supplementary material reports the same design estimated by GMM after forward orthogonal deviations with lagged levels as instruments \citep{arellano1995another}. Its sets cover the true response at least as often as their nominal level and are wider than those of the corrected within estimator, because lagged levels are weak instruments for the transformed lags when the unit effects are large relative to the innovations \citep{blundell1998initial}.}

With a weak instrument and an endogenous policy, the point estimate is strongly biased. In the weaker design, the Wald interval covers the true response less often than its nominal level at every horizon, and its intervals are much wider than the range of plausible responses. The AR sets keep coverage close to the nominal level at every horizon and for both strengths of the instrument. In the weaker design, more than half of the AR sets are unbounded, which is the correct report when the data contain little information on the response: any confidence set with correct coverage under weak identification must be unbounded with positive probability \citep{dufour1997some}. Sets clustered by county alone fail almost completely in this design (Table~S4 of the supplementary material), since the endogeneity operates through the common shocks.

Table~S5 of the supplementary material reports the calibrated design with a more persistent aggregate shift, $\rho_g=0.67$ and $\rho_g=0.9$, which makes the instrument correlated with past policy shocks. The first-stage statistic falls, and the two-way clustered AR sets with the jackknife correction remain close to their nominal coverage.

\IfFileExists{tables/table_cns_mc_clr.tex}{\begin{table}[ht]\centering\footnotesize\setlength{\tabcolsep}{4pt}
\begin{tabular}{llr rrr rrr}\toprule
 & & & \multicolumn{3}{c}{Coverage (\%)} & \multicolumn{3}{c}{Rejection of zero (\%)}\\
\cmidrule(lr){4-6}\cmidrule(lr){7-9}
Design & $h$ & $F$ & AR, $k=1$ & AR, $k=4$ & CLR, $k=4$ & AR, $k=1$ & AR, $k=4$ & CLR, $k=4$\\\midrule
Calibrated, $k=4$ & 0 & 24.5 & 94.6 & 93.8 & 95.6 & 15.2 & 11.2 & 13.8\\
 & 4 &  & 95.0 & 94.4 & 97.2 & 98.2 & 87.4 & 93.0\\
 & 7 &  & 94.8 & 95.2 & 96.8 & 78.2 & 47.6 & 61.8\\
 & 12 &  & 94.8 & 94.2 & 96.6 & 52.0 & 20.2 & 31.4\\
\addlinespace
Weak, $c=0.4$, $k=4$ & 0 & 4.8 & 93.4 & 94.4 & 96.4 & 6.6 & 6.4 & 5.0\\
 & 4 &  & 94.6 & 96.0 & 97.4 & 53.0 & 27.6 & 33.8\\
 & 7 &  & 93.8 & 95.4 & 97.0 & 31.4 & 15.2 & 18.4\\
 & 12 &  & 94.0 & 95.4 & 97.8 & 22.4 & 10.4 & 10.4\\
\addlinespace
\bottomrule\end{tabular}
\caption{Several instruments: the shift-share instrument and the exposure times the shift lagged one to three quarters ($k=4$). Coverage of nominal 95\% sets for the cumulated response at horizon $h$, and rejection rates of the null that the response is zero. AR, $k=4$, and CLR, $k=4$: Anderson-Rubin and conditional likelihood ratio tests with all instruments; AR, $k=1$: Anderson-Rubin test with the relevant instrument alone. Within estimator with the half-panel jackknife correction, two-way clustering. $F$: median first-stage statistic with all instruments, two-way clustered. 500 replications per design.\label{tab:mc_clr}}
\end{table}
}{}

With several instruments, the AR and CLR sets that use all four instruments have coverage close to or above the nominal level, the CLR sets being slightly conservative (Table~\ref{tab:mc_clr}). Adding instruments that carry no information on the policy innovation lowers the first-stage statistic and the power of the AR test, whose degrees of freedom grow with the number of instruments. The CLR test recovers part of this loss: after impact, it rejects false values of the response more often than the AR test with all four instruments at almost every horizon and in both designs. It does not reach the power of the AR test that uses the relevant instrument alone, which is not feasible when the researcher does not know which instrument is relevant.

\section{Stock market wealth and local labor markets\label{sec:Application}}

The effects of changes in the stock market on the real economy are a subject of debate among economists \citep{poterba2000stock,dimaggio2020stock,chodorowreich2021}. This is partly because aggregate time series estimates may possess a form of anticipation bias: stock prices anticipate changes in future income and so move with consumption even without a causal effect of wealth \citep{poterba2000stock}. Household level data show that spending does respond to capital gains on stocks, with the response concentrated among wealthier households \citep{dimaggio2020stock}. I revisit the study of \citet{chodorowreich2021}, who exploit the uneven geographic distribution of stock wealth across US counties. They use an observed feature of the economy: when stock prices rise, counties where households hold more stock receive a larger wealth gain, and if households spend part of it locally, employment and payroll in local nontradable sectors increase, an approach to local demand that follows \citet{mian2014explains}. They estimate the response of county payroll with local projections instrumented by the interaction of lagged county wealth with the aggregate market return, and find that payroll increases and that the effect builds over the two years after the shock.

I use the PVAR of Section~\ref{sec:Identification} and the estimation and inference of Section~\ref{sec:Inference}. The PVAR delivers the whole path of the response from a single instrumental variable estimate of the impact vector and the estimated dynamics, and Theorem~\ref{thm:PVAR-LATE} gives the impact response its interpretation as a $\mu$-LATE. The data are the county panel of the replication package of \citet{chodorowreich2021}. The panel is quarterly and covers $N=2{,}897$ counties over $T=92$ quarters between 1992 and 2015. The policy variable $W_{it}$ is the stock wealth shock of county $i$ in quarter $t$, the capital gain on the stock wealth held by its residents, scaled by annual county labor income. The outcome is payroll in the nontradable sector, and it enters the PVAR as its quarterly growth, $\Delta Y_{it}$, with $Y_{it}$ one hundred times the logarithm of payroll. I report the cumulated response of $\Delta Y_{it}$, which is the response of the level of payroll relative to the quarter before the shock. This is the outcome of the local projections of the original study, and the specification in growth rates allows a permanent change in wealth to have a permanent effect on the level of payroll.\footnote{In a PVAR with the level of payroll, the largest root of the estimated system after the bias correction is at or slightly above one, because county payroll is highly persistent in levels. The specification in growth rates keeps the system well inside the stationary region, as Assumption~\ref{assu:Regularity}(i) requires.} A unit change in the policy is a wealth gain equal to one percent of annual labor income, and the responses are in basis points of nontradable payroll. The instrument is
\[
Z_{it}=S_{i,t-8}\,R^m_t,
\]
the stock wealth of county $i$ eight quarters earlier, relative to labor income, multiplied by the aggregate stock market return in quarter $t$, as in the main specification of \citet{chodorowreich2021}. The instrument has a shift-share form: it interacts an aggregate shift, the market return, with a county exposure, the lagged stock wealth \citep{bartik1991benefits,goldsmith2020bartik,borusyak2022quasi}. The estimated model is
\[
x_{it}=\alpha_i+\delta_{c(i),t}+\sum_{l=1}^{8}\Phi_l\,x_{i,t-l}+\Lambda\,q_{it}+\tilde{x}_{it},\qquad x_{it}=(W_{it},\Delta Y_{it})',
\]
where $\alpha_i$ are county effects, $\delta_{c(i),t}$ are effects for each state and quarter, with $c(i)$ the state of county $i$, and $q_{it}$ collects the controls of the original study: the lagged employment predicted by national industry trends, and the interactions of lagged wealth with lagged aggregate returns. The effects for each state and quarter absorb aggregate and regional disturbances, including the financial crisis and the recessions of the sample, so no further dummies are needed. The eight lags match the eight lags of the shock in the original distributed lag specification. Counties are weighted by their population in 2010, as in the original study. The estimator is the within estimator with the half-panel jackknife correction, the covariance is clustered by county and by quarter, and the sets for the cumulated response are those of Section~\ref{sec:Inference}.

I now discuss why the assumptions of Section~\ref{sec:Assumptions} are plausible in this study. First, the aggregate return is close to unpredictable and is not caused by the conditions of any single county, and the exposure is measured eight quarters before the return, so that the instrument is not driven by contemporaneous local shocks. The effects for each state and quarter remove regional shocks, and the controls remove the differential exposure of wealthier counties to the aggregate conditions captured by past returns. In this sense, it is reasonable to assume that the shift-share instrument is exogenous to local labor market conditions, and affects them only through the stock-market spending channel. Moreover, monotonicity holds because the exposure is nonnegative: a rise in the market raises the wealth of every county that holds stocks. Because the instrument has the shift-share form and the model includes effects for each state and quarter, the $\mu$-LATE weights counties in proportion to the squared deviation of their exposure from the average exposure of their state, and the population weights add the size of the county. The estimand is therefore an average of county effects that gives more weight to populous counties whose stock wealth differs most from that of their state. Finally, Assumption~\ref{assu:Shock}(i) requires the identified policy shock to be serially uncorrelated.\footnote{Table~\ref{tab:cns_shock} reports the autocorrelation of the policy shock recovered from the estimated PVAR and of the instrument, and the correlation of the instrument with the lagged residuals of each equation, with $t$-statistics clustered by county and quarter. None of the autocorrelations of the policy shock is significantly different from zero, and the instrument is uncorrelated with the past innovations of payroll. The reduced-form residuals are instead serially correlated, but with more than two hundred thousand observations clustered by county even small correlations have large $t$-statistics, and Assumption~\ref{assu:Shock}(i) concerns the policy shock identified by the instrument; the two-way clustered covariance accounts for the serial correlation of the moment contributions within each county. Table~S1 of the supplementary material reports, for each specification, the first-stage statistics, the residual autocorrelation and the largest modulus of the eigenvalues of the companion matrix, which lies inside the unit circle in every specification.}

\begin{table}[ht]\centering\footnotesize
\begin{tabular}{lcccc}\toprule
 & $l=1$ & $l=2$ & $l=3$ & $l=4$\\\midrule
Identified policy shock, $\hat\varepsilon_t$ on $\hat\varepsilon_{t-l}$ & 0.133 & 0.078 & 0.202 & 0.032\\
 & (1.20) & (0.83) & (1.82) & (0.28)\\[2pt]
Instrument, $Z_t$ on $Z_{t-l}$ & 0.280 & 0.038 & 0.026 & -0.009\\
 & (1.97) & (0.25) & (0.16) & (-0.08)\\[2pt]
Instrument and wealth-shock residual, $Z_t$ on $\tilde{x}_{1,t-l}$ & 0.143 & 0.020 & 0.095 & -0.140\\
 & (1.34) & (0.20) & (0.97) & (-1.35)\\[2pt]
Instrument and payroll residual, $Z_t$ on $\tilde{x}_{2,t-l}$ & -0.001 & 0.005 & -0.002 & -0.008\\
 & (-0.15) & (0.88) & (-0.52) & (-1.49)\\[2pt]
\bottomrule\end{tabular}
\caption{Serial correlation of the identified policy shock and of the instrument, baseline specification. Each entry is the correlation between the current and the lagged series, pooled over counties and quarters, with the $t$-statistic of the corresponding regression, two-way clustered by county and quarter, in parentheses. The policy shock is recovered as $\hat\Gamma'\hat\Sigma^{-1}\tilde{x}_t$.\label{tab:cns_shock}}
\end{table}

Figure~\ref{fig:cns_main} reports the response of nontradable payroll to a wealth gain of one percent of annual labor income. The response is small on impact, builds within the first two years, and then settles at a permanently higher level of payroll. The AR sets exclude zero at every horizon from the first quarter after the shock to the end of the four years considered. Clustering by quarter alone gives sets of a similar width, which indicates that the common shocks dominate the dependence of the moment contributions, as expected with a shift-share instrument. The figure also reports the local projection estimates of the original study, which are the estimates that would be obtained without the PVAR. They lie inside the AR sets of the PVAR at every horizon, and the PVAR traces a smoother path, because it estimates all horizons from the same impact vector and dynamics.

\begin{figure}[ht]
\centering
\includegraphics[width=\textwidth]{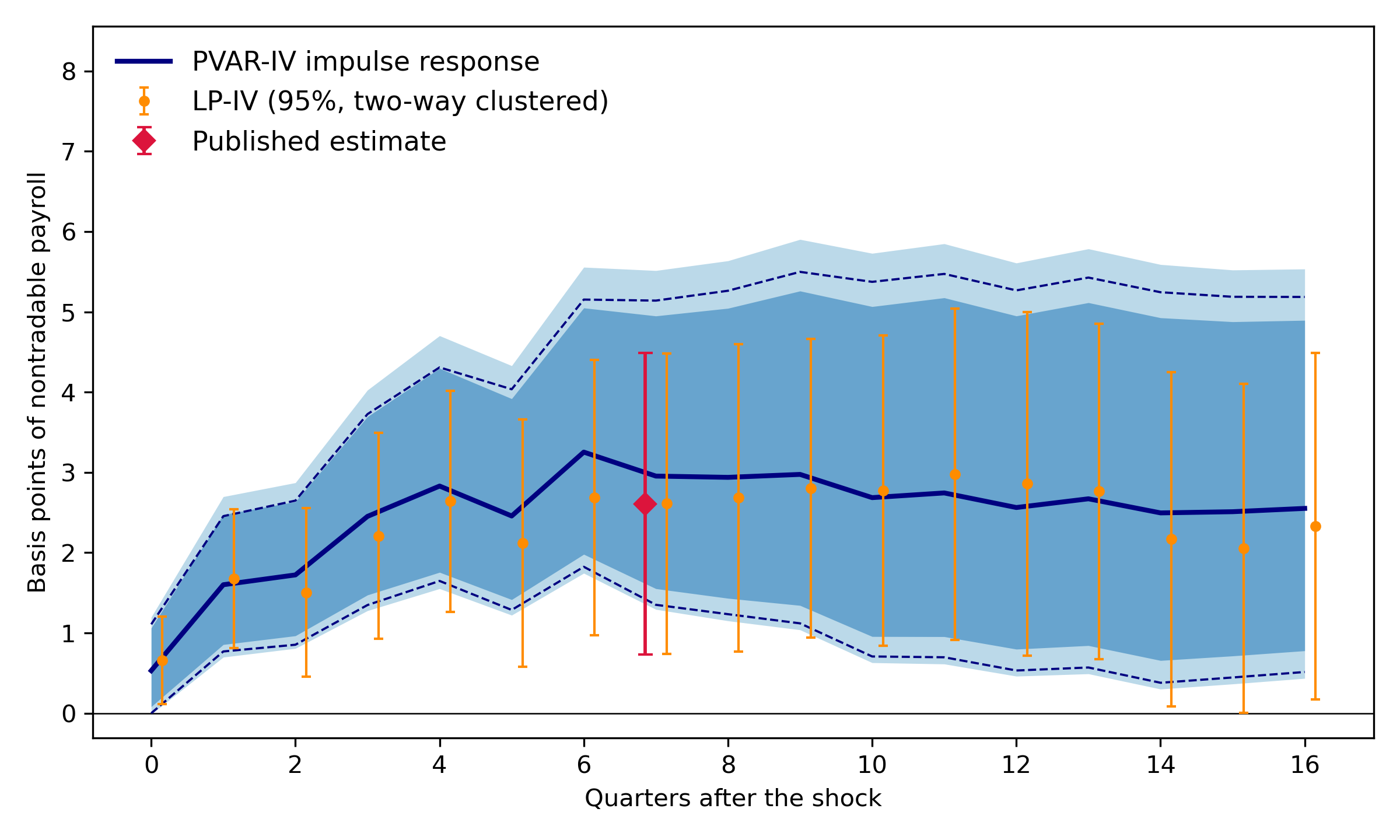}
\caption{Cumulated response of nontradable payroll, in basis points, to a stock wealth shock equal to one percent of annual labor income. PVAR-IV in the stock wealth shock and payroll growth, estimated by the within estimator with the half-panel jackknife correction, eight lags, county and state-by-quarter effects, population weights. Shaded areas: 90\% and 95\% Anderson-Rubin sets clustered by county and quarter. Dashed lines: 95\% Anderson-Rubin set clustered by quarter. Dots: local projection IV estimates of \citet{chodorowreich2021} with 95\% intervals clustered by county and quarter. Diamond: published estimate at seven quarters.\label{fig:cns_main}}
\end{figure}

Figure~\ref{fig:cns_robust} and Table~S2 of the supplementary material report the response under alternative specifications. Each of those specifications is similar to the original Local Projection study and the headline Panel SVAR over the first three years. With eight lags, the specification without population weights and the one with the alternative instrument of the original study, which measures the wealth exposure with fixed dividend yields, give the same shape, a response that builds within the first two years and then settles, and their sets exclude zero at almost every horizon. With twelve and sixteen lags, the latter selected by the Bayesian information criterion, the response over the first three years is similar to the baseline and its sets exclude zero, while at longer horizons the cumulated response keeps rising and the sets widen.\footnote{One explanation is that with more than two hundred thousand observations the penalty of the criterion is small relative to the gain in fit, making the criterion select the largest lags possible.} The level of payroll reached after three years is therefore less precisely determined than its path over the first three years. Importantly, the figure also reports the uncorrected within estimator, which is close to the baseline, and the GMM estimator after forward orthogonal deviations, which is lower at longer horizons and less precise, in line with the weakness of lagged levels as instruments discussed in Section~\ref{sec:MonteCarlo}.

\begin{figure}[ht]
\centering
\includegraphics[width=\textwidth]{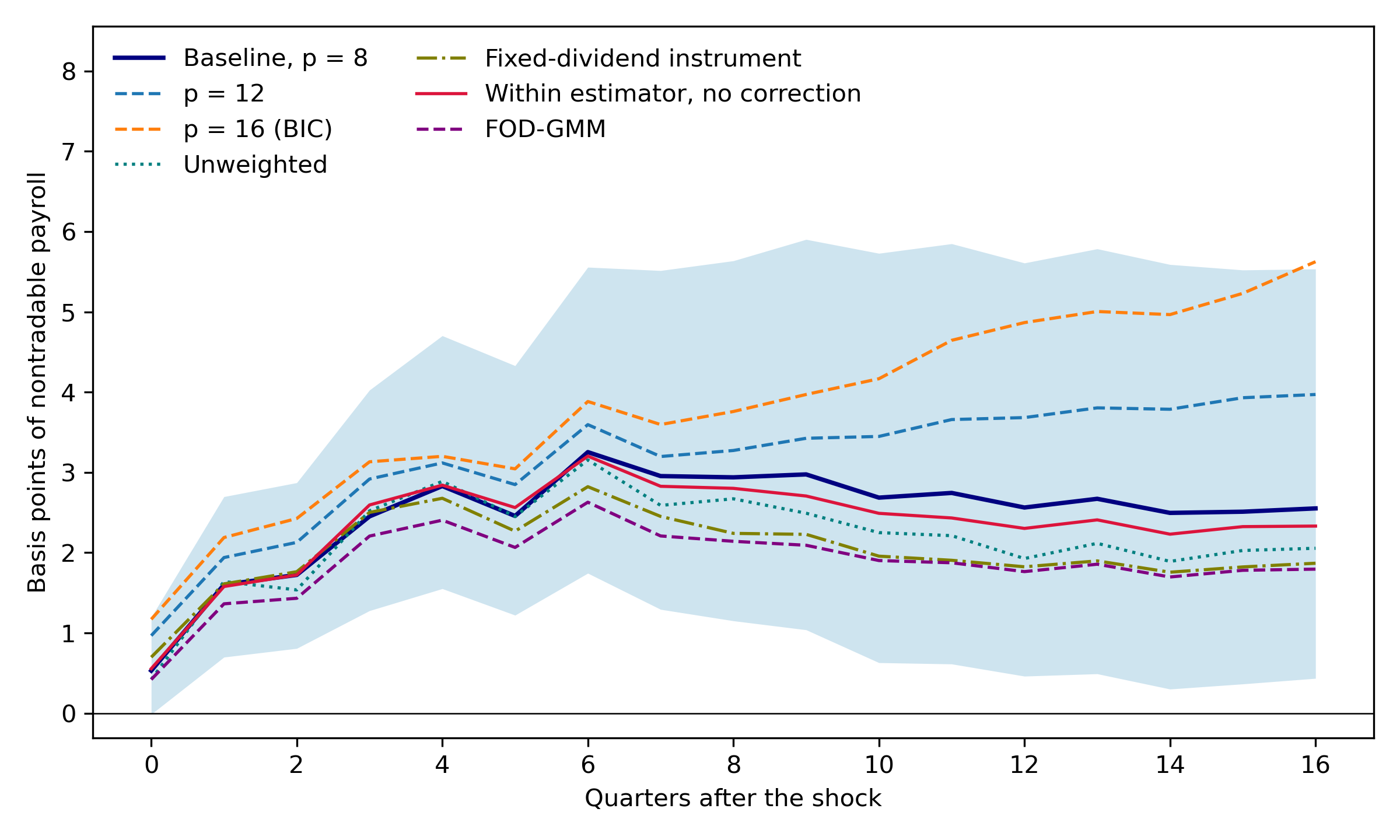}
\caption{Cumulated response of nontradable payroll, in basis points, to a stock wealth shock equal to one percent of annual labor income, under alternative specifications. Shaded area: 95\% Anderson-Rubin set of the baseline, clustered by county and quarter. Lines: baseline with eight lags; twelve and sixteen lags, the latter selected by the Bayesian information criterion; without population weights; instrument with fixed dividend yields; within estimator without correction; GMM after forward orthogonal deviations with lagged levels as instruments. All specifications other than the last two use the half-panel jackknife correction.\label{fig:cns_robust}}
\end{figure}

\section{Conclusions\label{sec:Conclusions}}

This paper studied panel vector autoregressions identified with an external instrument. Its most relevant contribution is in the causal spectrum. I show that the instrumental variable estimator of the impact response is a weighted average of unit-level effects plus a bias term. Under independence, exclusion, monotonicity and linearity the bias vanishes and the impact response identifies the $\mu$-LATE, with nonnegative weights. Responses at longer horizons additionally require a serially uncorrelated policy shock and an instrument uncorrelated with the other shocks at every date. I estimated the dynamics with a half-panel jackknife correction of the within estimator and constructed Anderson-Rubin sets and CLR sets for impulse responses, cumulated responses and cumulative multipliers with two-way clustering, which are valid when the instrument is weak and when common shocks make units dependent within a period.

Monte Carlo simulations calibrated to the application showed that clustering by unit alone fails with a shift-share instrument, that the jackknife correction keeps the coverage of the sets close to the nominal level at longer horizons, that the Anderson-Rubin sets remain valid with a weak instrument and with a persistent one, and that with several instruments the conditional likelihood ratio test has more power than the Anderson-Rubin test. In the application to stock market wealth and local labor markets in US counties, the response of nontradable payroll builds within the first two years after a wealth gain and then settles at a permanently higher level, consistently with the local projections of the original study at every horizon.

\section*{Supplementary materials}
The supplementary material at the end of this document contains the proofs, the extension to several instruments, and additional results for the application and the Monte Carlo simulations. The code that reproduces every result, with instructions to obtain the data, is available at \url{https://github.com/Raimondo-P/Identification-estimation-and-inference-in-Panel-Vector-Autoregressions-using-external-instruments} under the Creative Commons Attribution license, and through the author's webpage, \url{https://raimondo-p.github.io}.

\section*{Disclosure statement}
The author reports there are no competing interests to declare.

\section*{Funding}
No funding was obtained for this work.

\section*{Declaration of generative AI use}
The first version of this paper on arXiv was prepared without the use of generative AI. In preparing the present version, I used Claude (Anthropic) to address concerns raised in seminars and conferences. Claude assisted in developing the extensions of the theory (the bias decomposition, the inference results under two-way dependence and their proofs, and the extension to several instruments), in designing and implementing the Monte Carlo simulations, in developing the Python package and the code for the empirical application, and in drafting and editing the text. All modeling and econometric choices were made by the author. I verified the derivations, checked the code and reproduced all results with it, and take full responsibility for the content of the paper.

\bibliographystyle{apalike}
\bibliography{bibliography}

@article{angrist1995twostage,
  author  = {Angrist, Joshua D. and Imbens, Guido W.},
  title   = {Two-Stage Least Squares Estimation of Average Causal Effects in Models with Variable Treatment Intensity},
  journal = {Journal of the American Statistical Association},
  year    = {1995},
  volume  = {90},
  number  = {430},
  pages   = {431--442}
}

@article{nickell1981biases,
  author  = {Nickell, Stephen},
  title   = {Biases in Dynamic Models with Fixed Effects},
  journal = {Econometrica},
  year    = {1981},
  volume  = {49},
  number  = {6},
  pages   = {1417--1426}
}

@article{borusyak2022quasi,
  author  = {Borusyak, Kirill and Hull, Peter and Jaravel, Xavier},
  title   = {Quasi-Experimental Shift-Share Research Designs},
  journal = {Review of Economic Studies},
  year    = {2022},
  volume  = {89},
  number  = {1},
  pages   = {181--213}
}

@article{adao2019shift,
  author  = {Ad{\~a}o, Rodrigo and Kolesár, Michal and Morales, Eduardo},
  title   = {Shift-Share Designs: Theory and Inference},
  journal = {Quarterly Journal of Economics},
  year    = {2019},
  volume  = {134},
  number  = {4},
  pages   = {1949--2010}
}

@article{rameyzubairy2018,
  author  = {Ramey, Valerie A. and Zubairy, Sarah},
  title   = {Government Spending Multipliers in Good Times and in Bad: Evidence from {US} Historical Data},
  journal = {Journal of Political Economy},
  year    = {2018},
  volume  = {126},
  number  = {2},
  pages   = {850--901}
}

@article{arellano1995another,
  author  = {Arellano, Manuel and Bover, Olympia},
  title   = {Another Look at the Instrumental Variable Estimation of Error-Components Models},
  journal = {Journal of Econometrics},
  year    = {1995},
  volume  = {68},
  number  = {1},
  pages   = {29--51}
}

@article{forastiere2021identification,
  title={Identification and estimation of treatment and interference effects in observational studies on networks},
  author={Forastiere, Laura and Airoldi, Edoardo M and Mealli, Fabrizia},
  journal={Journal of the American Statistical Association},
  volume={116},
  number={534},
  pages={901--918},
  year={2021},
  publisher={Taylor \& Francis}
}

@article{aronow2017estimating,
  title={Estimating average causal effects under general interference, with application to a social network experiment},
  author={Aronow, Peter M and Samii, Cyrus},
journal={Annals of Applied Statistics},
  year={2017}
}

@book{granger2014forecasting,
  title={Forecasting economic time series},
  author={Granger, Clive William John and Newbold, Paul},
  year={2014},
  publisher={Academic press}
}

@book{bartik1991benefits,
  title={Who benefits from state and local economic development policies?},
  author={Bartik, Timothy J},
  year={1991},
  publisher={WE Upjohn Institute for Employment Research}
}

@article{goldsmith2020bartik,
  title={Bartik instruments: What, when, why, and how},
  author={Goldsmith-Pinkham, Paul and Sorkin, Isaac and Swift, Henry},
  journal={American Economic Review},
  volume={110},
  number={8},
  pages={2586--2624},
  year={2020}
}

@article{andrews2006optimal,
  title={Optimal two-sided invariant similar tests for instrumental variables regression},
  author={Andrews, Donald WK and Moreira, Marcelo J and Stock, James H},
  journal={Econometrica},
  volume={74},
  number={3},
  pages={715--752},
  year={2006},
  publisher={Wiley Online Library}
}

@article{mikusheva2010robust,
  title={Robust confidence sets in the presence of weak instruments},
  author={Mikusheva, Anna},
  journal={Journal of Econometrics},
  volume={157},
  number={2},
  pages={236--247},
  year={2010},
  publisher={Elsevier}
}

@misc{stock2002testing,
  title={Testing for weak instruments in linear IV regression},
  author={Stock, James H and Yogo, Motohiro},
  year={2002},
  publisher={National Bureau of Economic Research Cambridge, Mass., USA}
}

@article{antonelli2023principal,
  title={Principal stratification with continuous treatments and continuous post-treatment variables},
  author={Antonelli, Joseph and Mealli, Fabrizia and Beck, Brenden and Mattei, Alessandra},
  journal={arXiv preprint arXiv:2309.14486},
  year={2023}
}

@article{pala2024pvarcontrol,
  title={The causal interpretation of panel vector autoregressions},
  author={Pala, Raimondo},
  journal={arXiv preprint arXiv:2510.23540},
  year={2025}
}

@article{gertler2015monetary,
  title={Monetary policy surprises, credit costs, and economic activity},
  author={Gertler, Mark and Karadi, Peter},
  journal={American Economic Journal: Macroeconomics},
  volume={7},
  number={1},
  pages={44--76},
  year={2015}
}

@article{menchettibojinov2022,
author = {Fiammetta Menchetti and Iavor Bojinov},
title = {{Estimating the effectiveness of permanent price reductions for competing products using multivariate Bayesian structural time series models}},
volume = {16},
journal = {The Annals of Applied Statistics},
number = {1},
publisher = {Institute of Mathematical Statistics},
pages = {414 -- 435},
year = {2022},
doi = {10.1214/21-AOAS1498},
URL = {https://doi.org/10.1214/21-AOAS1498}
}

@article{bojinovshephard2019,
  title={Time series experiments and causal estimands: exact randomization tests and trading},
  author={Bojinov, Iavor and Shephard, Neil},
  journal={Journal of the American Statistical Association},
  year={2019},
  publisher={Taylor \& Francis}
}

@article{menchetticipollinimealli2022,
    author = {Menchetti, Fiammetta and Cipollini, Fabrizio and Mealli, Fabrizia},
    title = "{Combining counterfactual outcomes and ARIMA models for policy evaluation}",
    journal = {The Econometrics Journal},
    volume = {26},
    number = {1},
    pages = {1-24},
    year = {2022},
    month = {09},
    issn = {1368-4221},
    doi = {10.1093/ectj/utac024},
    url = {https://doi.org/10.1093/ectj/utac024},
    eprint = {https://academic.oup.com/ectj/article-pdf/26/1/1/48597685/utac024.pdf},
}

@article {RambachanSheppard2021,
	title = {When do common time series estimands have nonparametric causal meaning?},
        journal = {Working Paper},
	year = {2021},
        month = {10},
	author = {Ashesh Rambachan and Neil Shephard},
}

@article {AndrewsStockSun2019,
	title = {Weak Instruments in IV Regression: Theory and Practice},
	journal = {Annual Review of Economics},
	volume = {11},
	year = {2019},
	month = {2019},
	pages = {727-753},
	author = {Isaiah Andrews and James Stock and Liyang Sun}
}

@article{StockWatson2018,
author = {Stock, James H. and Watson, Mark W.},
title = {Identification and Estimation of Dynamic Causal Effects in Macroeconomics Using External Instruments},
journal = {The Economic Journal},
volume = {128},
number = {610},
pages = {917-948},
doi = {https://doi.org/10.1111/ecoj.12593},
url = {https://onlinelibrary.wiley.com/doi/abs/10.1111/ecoj.12593},
eprint = {https://onlinelibrary.wiley.com/doi/pdf/10.1111/ecoj.12593},
year = {2018}
}

@article{StockStaiger1997,
 ISSN = {00129682, 14680262},
 URL = {http://www.jstor.org/stable/2171753},
 author = {Douglas Staiger and James H. Stock},
 journal = {Econometrica},
 number = {3},
 pages = {557--586},
 publisher = {[Wiley, Econometric Society]},
 title = {Instrumental Variables Regression with Weak Instruments},
 urldate = {2023-05-09},
 volume = {65},
 year = {1997}
}

@article{Anderson1949,
author = {T. W. Anderson and Herman Rubin},
title = {Estimation of the Parameters of a Single Equation in a Complete System of Stochastic Equations},
volume = {20},
journal = {The Annals of Mathematical Statistics},
number = {1},
publisher = {Institute of Mathematical Statistics},
pages = {46 -- 63},
year = {1949},
doi = {10.1214/aoms/1177730090},
URL = {https://doi.org/10.1214/aoms/1177730090}
}

@Article{Nakamura2018a,
  author    = {Emi Nakamura and J{\'{o}}n Steinsson},
  journal   = {Journal of Economic Perspectives},
  title     = {Identification in Macroeconomics},
  year      = {2018},
  month     = {aug},
  number    = {3},
  pages     = {59--86},
  volume    = {32},
  doi       = {10.1257/jep.32.3.59},
  publisher = {American Economic Association},
}

@Article{Olea2021,
  author    = {Jos{\'{e}} L. Montiel Olea and James H. Stock and Mark W. Watson},
  journal   = {Journal of Econometrics},
  title     = {Inference in Structural Vector Autoregressions identified with an external instrument},
  year      = {2021},
  month     = {nov},
  number    = {1},
  pages     = {74--87},
  volume    = {225},
  doi       = {10.1016/j.jeconom.2020.05.014},
  publisher = {Elsevier {BV}},
}

@Article{Angrist2011,
  author    = {Joshua D. Angrist and Guido M. Kuersteiner},
  journal   = {The Review of Economics and Statistics},
  title     = {CAUSAL EFFECTS OF MONETARY SHOCKS: SEMIPARAMETRIC CONDITIONAL INDEPENDENCE TESTS WITH A MULTINOMIAL PROPENSITY SCORE},
  year      = {2011},
  issn      = {00346535, 15309142},
  number    = {3},
  pages     = {725--747},
  volume    = {93},
  publisher = {The MIT Press},
  url       = {http://www.jstor.org/stable/23016073},
  urldate   = {2023-04-02},
}

@Article{Mertens2013,
  author    = {Mertens, Karel and Ravn, Morten},
  journal   = {American Economic Review},
  title     = {The Dynamic Effects of Personal and Corporate Income Tax Changes in the United States},
  year      = {2013},
  month     = {6},
  number    = {4},
  pages     = {1212-1247},
  volume    = {103},
  date      = {2013-06-01},
  day       = {1},
  doi       = {10.1257/aer.103.4.1212},
  publisher = {American Economic Association},
}

@Article{Mertens2014,
  author    = {Mertens, Karel and Ravn, Morten},
  journal   = {Journal of Monetary Economics},
  title     = {A reconciliation of SVAR and narrative estimates of tax multipliers},
  year      = {2014},
  month     = {12},
  pages     = {S1-S19},
  volume    = {68},
  date      = {2014-12},
}

@Article{Nakamura2014,
  author    = {Nakamura, Emi and Steinsson, Jon},
  journal   = {American Economic Review},
  title     = {Fiscal Stimulus in a Monetary Union: Evidence from US Regions},
  year      = {2014},
  month     = {3},
  number    = {3},
  pages     = {753-792},
  volume    = {104},
  date      = {2014-03-01},
  day       = {1},
  publisher = {American Economic Association},
}

@Article{Rubin1980,
  author    = {Rubin, Donald},
  journal   = {Journal of the American Statistical Association},
  title     = {Randomization Analysis of Experimental Data: The Fisher Randomization Test Comment},
  year      = {1980},
  month     = {9},
  number    = {371},
  pages     = {591},
  volume    = {75},
  date      = {1980-09},
  doi       = {10.2307/2287653},
  publisher = {JSTOR},
}

@article{moreira2003conditional,
  author  = {Moreira, Marcelo J.},
  title   = {A Conditional Likelihood Ratio Test for Structural Models},
  journal = {Econometrica},
  year    = {2003},
  volume  = {71},
  number  = {4},
  pages   = {1027--1048}
}

@article{kleibergen2005testing,
  author  = {Kleibergen, Frank},
  title   = {Testing Parameters in {GMM} without Assuming that They Are Identified},
  journal = {Econometrica},
  year    = {2005},
  volume  = {73},
  number  = {4},
  pages   = {1103--1123}
}

@article{andrews2016conditional,
  author  = {Andrews, Isaiah and Mikusheva, Anna},
  title   = {Conditional Inference with a Functional Nuisance Parameter},
  journal = {Econometrica},
  year    = {2016},
  volume  = {84},
  number  = {4},
  pages   = {1571--1612}
}

@article{moreira2019optimal,
  author  = {Moreira, Humberto and Moreira, Marcelo J.},
  title   = {Optimal Two-Sided Tests for Instrumental Variables Regression with Heteroskedastic and Autocorrelated Errors},
  journal = {Journal of Econometrics},
  year    = {2019},
  volume  = {213},
  number  = {2},
  pages   = {398--433},
  doi     = {10.1016/j.jeconom.2019.04.038}
}

@article{chodorowreich2021,
  author  = {Chodorow-Reich, Gabriel and Nenov, Plamen T. and Simsek, Alp},
  title   = {Stock Market Wealth and the Real Economy: A Local Labor Market Approach},
  journal = {American Economic Review},
  year    = {2021},
  volume  = {111},
  number  = {5},
  pages   = {1613--1657},
  doi     = {10.1257/aer.20200208}}

@article{dhaene2015split,
  author  = {Dhaene, Geert and Jochmans, Koen},
  title   = {Split-Panel Jackknife Estimation of Fixed-Effect Models},
  journal = {Review of Economic Studies},
  year    = {2015},
  volume  = {82},
  number  = {3},
  pages   = {991--1030}
}

@article{dufour1997some,
  author  = {Dufour, Jean-Marie},
  title   = {Some Impossibility Theorems in Econometrics with Applications to Structural and Dynamic Models},
  journal = {Econometrica},
  year    = {1997},
  volume  = {65},
  number  = {6},
  pages   = {1365--1387}
}

@article{blundell1998initial,
  author  = {Blundell, Richard and Bond, Stephen},
  title   = {Initial Conditions and Moment Restrictions in Dynamic Panel Data Models},
  journal = {Journal of Econometrics},
  year    = {1998},
  volume  = {87},
  number  = {1},
  pages   = {115--143}
}

@article{cameron2011robust,
  author  = {Cameron, A. Colin and Gelbach, Jonah B. and Miller, Douglas L.},
  title   = {Robust Inference with Multiway Clustering},
  journal = {Journal of Business \& Economic Statistics},
  year    = {2011},
  volume  = {29},
  number  = {2},
  pages   = {238--249}
}

@article{poterba2000stock,
  author  = {Poterba, James M.},
  title   = {Stock Market Wealth and Consumption},
  journal = {Journal of Economic Perspectives},
  year    = {2000},
  volume  = {14},
  number  = {2},
  pages   = {99--118}
}

@article{dimaggio2020stock,
  author  = {Di Maggio, Marco and Kermani, Amir and Majlesi, Kaveh},
  title   = {Stock Market Returns and Consumption},
  journal = {Journal of Finance},
  year    = {2020},
  volume  = {75},
  number  = {6},
  pages   = {3175--3219}
}

@article{mian2014explains,
  author  = {Mian, Atif and Sufi, Amir},
  title   = {What Explains the 2007--2009 Drop in Employment?},
  journal = {Econometrica},
  year    = {2014},
  volume  = {82},
  number  = {6},
  pages   = {2197--2223}
}

@article{hahnkuersteiner2002,
  author  = {Hahn, Jinyong and Kuersteiner, Guido},
  title   = {Asymptotically Unbiased Inference for a Dynamic Panel Model with Fixed Effects When Both $n$ and $T$ Are Large},
  journal = {Econometrica},
  year    = {2002},
  volume  = {70},
  number  = {4},
  pages   = {1639--1657}
}

@article{mackinnon2021wild,
  author  = {MacKinnon, James G. and Nielsen, Morten {\O}rregaard and Webb, Matthew D.},
  title   = {Wild Bootstrap and Asymptotic Inference with Multiway Clustering},
  journal = {Journal of Business \& Economic Statistics},
  year    = {2021},
  volume  = {39},
  number  = {2},
  pages   = {505--519}
}

@article{davezies2021empirical,
  author  = {Davezies, Laurent and D'Haultf{\oe}uille, Xavier and Guyonvarch, Yannick},
  title   = {Empirical Process Results for Exchangeable Arrays},
  journal = {Annals of Statistics},
  year    = {2021},
  volume  = {49},
  number  = {2},
  pages   = {845--862}
}

@article{imbens1994identification,
  author  = {Imbens, Guido W. and Angrist, Joshua D.},
  title   = {Identification and Estimation of Local Average Treatment Effects},
  journal = {Econometrica},
  year    = {1994},
  volume  = {62},
  number  = {2},
  pages   = {467--475}
}

@article{bound1995problems,
  author  = {Bound, John and Jaeger, David A. and Baker, Regina M.},
  title   = {Problems with Instrumental Variables Estimation When the Correlation Between the Instruments and the Endogenous Explanatory Variable Is Weak},
  journal = {Journal of the American Statistical Association},
  year    = {1995},
  volume  = {90},
  number  = {430},
  pages   = {443--450}
}

@article{jentsch2019dynamic,
  author  = {Jentsch, Carsten and Lunsford, Kurt G.},
  title   = {The Dynamic Effects of Personal and Corporate Income Tax Changes in the United States: Comment},
  journal = {American Economic Review},
  year    = {2019},
  volume  = {109},
  number  = {7},
  pages   = {2655--2678}
}

@article{li2024local,
  author  = {Li, Dake and Plagborg-M{\o}ller, Mikkel and Wolf, Christian K.},
  title   = {Local Projections vs. {VARs}: Lessons from Thousands of {DGPs}},
  journal = {Journal of Econometrics},
  year    = {2024},
  volume  = {244},
  number  = {2}
}

@article{plagborg2021local,
  author  = {Plagborg-M{\o}ller, Mikkel and Wolf, Christian K.},
  title   = {Local Projections and {VARs} Estimate the Same Impulse Responses},
  journal = {Econometrica},
  year    = {2021},
  volume  = {89},
  number  = {2},
  pages   = {955--980}
}

@article{hudgens2008toward,
  author  = {Hudgens, Michael G. and Halloran, M. Elizabeth},
  title   = {Toward Causal Inference With Interference},
  journal = {Journal of the American Statistical Association},
  year    = {2008},
  volume  = {103},
  number  = {482},
  pages   = {832--842}
}

@article{forastiere2016identification,
  author  = {Forastiere, Laura and Mealli, Fabrizia and VanderWeele, Tyler J.},
  title   = {Identification and Estimation of Causal Mechanisms in Clustered Encouragement Designs: Disentangling Bed Nets Using {B}ayesian Principal Stratification},
  journal = {Journal of the American Statistical Association},
  year    = {2016},
  volume  = {111},
  number  = {514},
  pages   = {510--525}
}

@article{forastiere2021exploring,
  author  = {Forastiere, Laura and Lattarulo, Patrizia and Mariani, Marco and Mealli, Fabrizia and Razzolini, Laura},
  title   = {Exploring Encouragement, Treatment, and Spillover Effects Using Principal Stratification, With Application to a Field Experiment on Teens' Museum Attendance},
  journal = {Journal of Business \& Economic Statistics},
  year    = {2021},
  volume  = {39},
  number  = {1},
  pages   = {244--258}
}

@article{forastiere2022estimating,
  author  = {Forastiere, Laura and Mealli, Fabrizia and Wu, Albert and Airoldi, Edoardo M.},
  title   = {Estimating Causal Effects under Network Interference with {B}ayesian Generalized Propensity Scores},
  journal = {Journal of Machine Learning Research},
  year    = {2022},
  volume  = {23},
  number  = {289},
  pages   = {1--61}
}

@article{hoshino2024causal,
  author  = {Hoshino, Tadao and Yanagi, Takahide},
  title   = {Causal Inference With Noncompliance and Unknown Interference},
  journal = {Journal of the American Statistical Association},
  year    = {2024},
  volume  = {119},
  number  = {548},
  pages   = {2869--2880}
}

@article{mei2026nickell,
  author  = {Mei, Ziwei and Sheng, Liugang and Shi, Zhentao},
  title   = {Nickell Bias in Panel Local Projection: Financial Crises Are Worse Than You Think},
  journal = {Journal of International Economics},
  year    = {2026}
}

\clearpage
\appendix
\setcounter{section}{0}\setcounter{equation}{0}\setcounter{table}{0}\setcounter{thm}{0}\setcounter{lem}{0}\setcounter{prop}{0}\setcounter{assumption}{0}
\renewcommand{\thesection}{S\arabic{section}}
\renewcommand{\thethm}{S\arabic{thm}}
\renewcommand{\thelem}{S\arabic{lem}}
\renewcommand{\theprop}{S\arabic{prop}}
\renewcommand{\theassumption}{S\arabic{assumption}}
\renewcommand{\thetable}{S\arabic{table}}
\renewcommand{\theequation}{S\arabic{equation}}
\begin{center}\Large Supplementary material\end{center}

This supplement contains the proofs of the results in the paper (Section~S1), the extension to several instruments (Section~S2), additional results for the application and the Monte Carlo simulations (Section~S3), a comparison with local projections (Section~S4), and the case of spillovers across units (Section~S5). Numbers of assumptions, theorems, lemmas, propositions, equations and tables without the prefix S refer to the main text.

\section{Proofs\label{sec:AppB}}

\begin{proof}[Proof of Lemma~\ref{lem:FWL}]
By the Frisch-Waugh-Lovell theorem, the coefficient on $W_{it}$ in the two-stage least squares regression with controls $\Omega_{it}$ equals the coefficient of the just-identified instrumental variable regression of $\tilde{Y}_{j,it}$ on $\tilde{W}_{it}$ with the instrument residualized on the regressors in $\Omega_{it}$. Let $R_{it}$ denote that residual of the instrument, so that $Z_{it}-R_{it}$ is a linear combination of the regressors in $\Omega_{it}$. The residuals $\tilde{Y}_{j,it}$ and $\tilde{W}_{it}$ are orthogonal in the sample to every linear combination of these regressors, so $\sum_{i,t}R_{it}\tilde{Y}_{j,it}=\sum_{i,t}Z_{it}\tilde{Y}_{j,it}$ and $\sum_{i,t}R_{it}\tilde{W}_{it}=\sum_{i,t}Z_{it}\tilde{W}_{it}$, and the coefficient is $\beta_j^{IV}$. The same orthogonality holds in population for the projection residuals $\tilde{x}_{it}$, so the two sample covariances, divided by the number of observations, converge to $\Gamma_j=\E[\tilde{Y}_{j,it}Z_{it}]$ and $\Gamma_1=\E[\tilde{W}_{it}Z_{it}]$ by the law of large numbers implied by Assumption~\ref{assu:Regularity}.
\end{proof}

\begin{proof}[Proof of Theorem~\ref{thm:decomp}]
I first introduce the notation. Write $Z=Z_{it}$, $\mu_Z=\mu_{Z,it}$, $\pi=\pi_{it}$, $\theta_j=\theta_{j,it}$, $W(0)=W_{it}(0)$ and $Y_j(0)=Y_{j,it}(0)$. By Assumption~\ref{assu:Regularity}(i), $\mu_Z$ is a linear combination of the regressors in $\Omega_{it}$, so $Z-\mu_Z$ is the residual of the instrument from the projection on these regressors. Let $A_j=\E[\theta_j\pi(Z-\mu_Z)^2]$ and $C=\E[\pi(Z-\mu_Z)^2]$, so that $\E[\omega_{it}\theta_{j,it}]=A_j/C$ and $\E[\omega_{it}]=1$.

By Assumption~\ref{assu:Linearity}, $W=W(0)+\pi Z=W(0)+\pi\mu_Z+\pi(Z-\mu_Z)$. Multiplying by $Z-\mu_Z$ and taking expectations,
\[
\E[W(Z-\mu_Z)]=\E\bigl[(Z-\mu_Z)\bigl(W(0)+\pi\mu_Z\bigr)\bigr]+\E[\pi(Z-\mu_Z)^2]=B_W+C.
\]
Likewise $Y_j=Y_j(0)+\theta_jW=Y_j(0)+\theta_jW(0)+\theta_j\pi\mu_Z+\theta_j\pi(Z-\mu_Z)$, so
\[
\E[Y_j(Z-\mu_Z)]=\E\bigl[(Z-\mu_Z)\bigl(Y_j(0)+\theta_jW(0)+\theta_j\pi\mu_Z\bigr)\bigr]+\E[\theta_j\pi(Z-\mu_Z)^2]=B_j+A_j.
\]
The residuals $\tilde{W}$ and $\tilde{Y}_j$ differ from $W$ and $Y_j$ by linear combinations of the regressors in $\Omega_{it}$, with which $Z-\mu_Z$ is uncorrelated by construction, while $\tilde{W}$ and $\tilde{Y}_j$ are uncorrelated with the linear combination $\mu_Z$. Hence $\Gamma_1=\E[\tilde{W}Z]=\E[\tilde{W}(Z-\mu_Z)]=\E[W(Z-\mu_Z)]=B_W+C$ and, in the same way, $\Gamma_j=\E[\tilde{Y}_jZ]=B_j+A_j$. By Lemma~\ref{lem:FWL}, $\beta_j^{IV}=\Gamma_j/\Gamma_1$. Since $C\neq0$ and $\Gamma_1\neq0$,
\begin{align*}
\beta_j^{IV}=\frac{A_j+B_j}{C+B_W}&=\frac{A_j}{C}+\frac{A_j+B_j}{C+B_W}-\frac{A_j}{C}\\
&=\frac{A_j}{C}+\frac{C(A_j+B_j)-A_j(C+B_W)}{C(C+B_W)}=\frac{A_j}{C}+\frac{B_j-(A_j/C)B_W}{C+B_W},
\end{align*}
which is $\E[\omega_{it}\theta_{j,it}]+\Delta_j$.
\end{proof}

\begin{proof}[Proof of Theorem~\ref{thm:PVAR-LATE}]
I use the notation of the previous proof and take all expectations first conditional on $\Omega_{it}$, given which $\mu_Z$ is fixed. By Assumption~\ref{assu:Independence}, $Z$ is independent of $(Y_j(0),W(0),\theta_j,\pi)$ given $\Omega_{it}$, and so is $Z-\mu_Z$, which has conditional mean zero by the definition of $\mu_Z$. Hence
\[
\E\bigl[(Z-\mu_Z)\bigl(W(0)+\pi\mu_Z\bigr)\mid\Omega_{it}\bigr]=\E\bigl[W(0)+\pi\mu_Z\mid\Omega_{it}\bigr]\,\E[(Z-\mu_Z)\mid\Omega_{it}]=0,
\]
and in the same way $\E[(Z-\mu_Z)(Y_j(0)+\theta_jW(0)+\theta_j\pi\mu_Z)\mid\Omega_{it}]=0$. Taking expectations over $\Omega_{it}$ gives $B_W=B_j=0$, so $\Delta_j=0$. For the weights, independence and the constant conditional variance give
\[
\E[\theta_j\pi(Z-\mu_Z)^2\mid\Omega_{it}]=\E[\theta_j\pi\mid\Omega_{it}]\,\sigma_Z^2,\qquad \E[\pi(Z-\mu_Z)^2\mid\Omega_{it}]=\E[\pi\mid\Omega_{it}]\,\sigma_Z^2,
\]
so $A_j=\sigma_Z^2\,\E[\theta_j\pi]$ and $C=\sigma_Z^2\,\E[\pi]$, and $C\neq0$ by Assumption~\ref{assu:Relevance}. Therefore $\beta_j^{IV}=A_j/C=\E[\theta_j\pi]/\E[\pi]=\mu\text{-LATE}_j$. Under Assumption~\ref{assu:Monotonicity}, $\pi_{it}$ has the same sign for all $(i,t)$, so $\pi_{it}/\E[\pi_{it}]\geq0$ and the weights are nonnegative. The equivalence with two-stage least squares is Lemma~\ref{lem:FWL}.

For the shift-share form, let $Z_{it}=s_ig_t$, let $\bar{s}$ be the average exposure, and suppose Assumption~\ref{assu:Independence} holds conditional on the unit, with $\var(\tilde{g}_t)=\sigma_g^2$ for the residualized shift. The period effects remove $\bar{s}g_t$, so $Z_{it}-\mu_{Z,it}=(s_i-\bar{s})\tilde{g}_t$ up to the adjustment for the unit mean, and the conditional variance of the instrument given unit $i$ is $(s_i-\bar{s})^2\sigma_g^2$. Repeating the argument unit by unit gives $\Gamma_j=\sigma_g^2N^{-1}\sum_i(s_i-\bar{s})^2\E[\theta_{j,it}\pi_{it}\mid i]$ and $\Gamma_1=\sigma_g^2N^{-1}\sum_i(s_i-\bar{s})^2\E[\pi_{it}\mid i]$, and their ratio is the expression in the text. With group-specific period effects the same argument applies within each group, with $\bar{s}$ the group average.
\end{proof}

\begin{proof}[Proof of Lemma~\ref{lem:IRF}]
(a) I first introduce the notation. For any variable $v_{it}$, let $\bar{v}_{i\cdot}$, $\bar{v}_{\cdot t}$ and $\bar{v}$ be its unit, period and overall means, and $\ddot{v}_{it}=v_{it}-\bar{v}_{i\cdot}-\bar{v}_{\cdot t}+\bar{v}$ the variable after removing the unit and period effects. Each transformed variable is a linear combination, with nonrandom coefficients, of the untransformed variable of all units at all dates.

By Assumption~\ref{assu:Shock}, $\tilde{x}_{it}=b\,\eta_{it}+\nu_{it}$ with a vector $b$ common to all units, so $\Gamma=\E[\tilde{x}_{it}Z_{it}]=b\,\E[\eta_{it}Z_{it}]+\E[\nu_{it}Z_{it}]=b\,\E[\eta_{it}Z_{it}]$, where the last term is zero by part (ii) with $j=i$ and $s=0$. The transformation is linear and $b$ is common to all units, so $\ddot{\tilde{x}}_{it}=b\,\ddot\eta_{it}+\ddot\nu_{it}$ and
\[
\ddot\Gamma=\E[\ddot{\tilde{x}}_{it}\ddot{Z}_{it}]=b\,\E[\ddot\eta_{it}\ddot{Z}_{it}]+\E[\ddot\nu_{it}\ddot{Z}_{it}].
\]
Expanding the two transformations, $\E[\ddot\nu_{it}\ddot{Z}_{it}]$ is a linear combination of terms $\E[\nu_{js}Z_{kr}]$ for units $j,k$ and dates $s,r$. Each of these is zero by Assumption~\ref{assu:Shock}(ii), which covers the shocks of every unit at every date. Hence $\ddot\Gamma=b\,\ddot\Gamma_1$ with $\ddot\Gamma_1=\E[\ddot\eta_{it}\ddot{Z}_{it}]$, which is nonzero by Assumption~\ref{assu:Shock}(iii). Since $b_1=1$, the first entry of $\ddot\Gamma$ is $\ddot\Gamma_1$, so $b=\ddot\Gamma/\ddot\Gamma_1$, and by Lemma~\ref{lem:FWL} applied to the transformed variables, $b_j=\ddot\Gamma_j/\ddot\Gamma_1$ is the probability limit of the instrumental variable coefficient of outcome $j$ with the transformed instrument.

(b) Let $\check{x}_{it}$ denote $x_{it}$ after removing $\alpha_i$. By model~\eqref{eq:main_PVAR} and the stability in Assumption~\ref{assu:Regularity}(i), $\check{x}_{it}=\sum_{s\geq0}\Phi^s\tilde{x}_{i,t-s}$. Because the innovations are serially uncorrelated and recoverable from current and past values of $x_{it}$ (Assumption~\ref{assu:Shock}(i)), this is the Wold representation, and $\Phi^s$ are its moving average coefficients. Substituting $\tilde{x}_{i,t+h-s}=b\,\eta_{i,t+h-s}+\nu_{i,t+h-s}$ gives
\[
\check{x}_{i,t+h}=\Phi^hb\,\eta_{it}+\sum_{s\neq h}\Phi^sb\,\eta_{i,t+h-s}+\sum_{s\geq0}\Phi^s\nu_{i,t+h-s},
\]
so the coefficient on $\eta_{it}$ is $\Phi^hb$. By Assumption~\ref{assu:Shock}(i), $\eta_{it}$ is uncorrelated with the other terms, so this coefficient also equals $\cov(\check{x}_{i,t+h},\eta_{it})/\var(\eta_{it})$.
\end{proof}

\begin{proof}[Proof of Proposition~\ref{prop:AR-coverage}]
I first introduce the notation. Let $\theta=(\mathrm{vec}(\Phi)',\Gamma')'$, let $\hat\theta$ be the uncorrected within estimator together with the moment computed from its residuals, and let $\hat\theta^{J}$ be the corrected estimator, with $\hat\theta^{(1)}$ and $\hat\theta^{(2)}$ the half-sample estimators. Let $\bar\psi=n^{-1}\sum_{i,t}\psi_{it}$ and $\bar\psi^{(k)}$ the same average over half $k$, so that $\bar\psi=\tfrac12(\bar\psi^{(1)}+\bar\psi^{(2)})$ up to the rounding of the split. Let $r_{NT}=n/\sqrt{v_{NT}}$, which lies between a multiple of $\sqrt{T}$ and $\sqrt{NT}$ by Assumption~\ref{assu:Regularity}(iii). Let $S$ be the Jacobian of $\theta$ with respect to the population moments of the transformed data, so that the estimators are smooth functions of sample moments with derivative $S$ at the population values. Finally, for fixed $(h,s,\lambda_0)$, let $a(\lambda_0)'=e_s'\Phi^h-\lambda_0e_1'$, $\hat a(\lambda_0)'=e_s'\hat\Phi^h-\lambda_0e_1'$, and $G_h$ the matrix of derivatives of $\mathrm{vec}(\Phi^h)$ with respect to $\mathrm{vec}(\Phi)$.

Step 1 (expansion of the uncorrected estimator). By Assumption~\ref{assu:Regularity}(i) and (ii), the sample second moments of the transformed regressors converge in probability to their population counterparts, and a first-order expansion of the least squares estimator and of the moment $\hat\Gamma$ around $\theta$ gives
\[
\hat\theta-\theta=S\,\bar\psi+\frac{B}{T}+R_{NT},
\]
where $B/T$ is the bias due to the removal of the unit means, of order $1/T$ \citep{nickell1981biases,hahnkuersteiner2002}, and $R_{NT}$ collects the second-order terms, of order $O_p(r_{NT}^{-2})+O(T^{-2})$. Period effects add no term of order $1/T$: the mean of period $t$ contains the innovation of unit $i$ at date $t$ with weight $1/N$, and that innovation is uncorrelated with the lags of unit $i$, which are predetermined, and with the lags and innovations of the other units at date $t$, which are uncorrelated with it outside the common component removed by the period effect. The bias in $\hat\Gamma$ affects only its scale (Lemma~\ref{lem:IRF}(a)) and is part of $B/T$.

Step 2 (the correction). Each half-sample estimator satisfies the same expansion with $T/2$ periods,
\[
\hat\theta^{(k)}-\theta=S\,\bar\psi^{(k)}+\frac{2B}{T}+R^{(k)}_{NT},\qquad k=1,2,
\]
because the bias depends on the number of periods per unit and the half samples have half of them, while $S$ is the same in both halves by the stationarity in Assumption~\ref{assu:Regularity}. Hence
\[
\hat\theta^{J}-\theta=2(\hat\theta-\theta)-\tfrac12\bigl(\hat\theta^{(1)}-\theta+\hat\theta^{(2)}-\theta\bigr)=S\,\bar\psi+R^{J}_{NT},
\]
with $R^J_{NT}=2R_{NT}-\tfrac12(R^{(1)}_{NT}+R^{(2)}_{NT})=O_p(r_{NT}^{-2})+O(T^{-2})$. Since $r_{NT}\leq\sqrt{NT}$ and $N/T\to\kappa$, $r_{NT}\,T^{-2}\leq\sqrt{N}\,T^{-3/2}\to0$, so $r_{NT}R^J_{NT}=o_p(1)$. This step uses the expansion of \citet{dhaene2015split} for the half-panel jackknife.

Step 3 (asymptotic normality and the covariance). By Assumption~\ref{assu:Regularity}(iii), $r_{NT}\,\bar\psi=v_{NT}^{-1/2}\sum_{i,t}\psi_{it}\xrightarrow{d}N(0,\Omega)$, so $r_{NT}(\hat\theta^J-\theta)\xrightarrow{d}N(0,S\Omega S')$. For the covariance, let $Q(\psi)=\sum_i(\sum_t\psi_{it})(\sum_t\psi_{it})'+\sum_t(\sum_i\psi_{it})(\sum_i\psi_{it})'-\sum_{i,t}\psi_{it}\psi_{it}'$ denote the two-way clustered sum, so that $v_{NT}^{-1}Q(\psi)\xrightarrow{p}\Omega$ by Assumption~\ref{assu:Regularity}(iii). The estimator uses $\hat\psi_{it}$, the contributions evaluated at $\hat\theta^J$. The contributions are linear in $\theta$, because the residuals are linear in $\Phi$ and the moment for $\Gamma$ is linear in $\Gamma$, so $\hat\psi_{it}=\psi_{it}+D_{it}\delta$ with $\delta=\hat\theta^J-\theta=O_p(r_{NT}^{-1})$ and a matrix $D_{it}$ of products of the data, whose second moments are bounded by Assumption~\ref{assu:Regularity}(i). Consider the sum clustered by period. Its change is bounded, by the Cauchy-Schwarz inequality, by terms of the form $\sum_t\|\sum_iD_{it}\delta\|^2$ and $\bigl(\sum_t\|\sum_i\psi_{it}\|^2\bigr)^{1/2}\bigl(\sum_t\|\sum_iD_{it}\delta\|^2\bigr)^{1/2}$. Since $\|\sum_iD_{it}\|^2\leq N\sum_i\|D_{it}\|^2$, the first term is $O_p(N^2T\,r_{NT}^{-2})=O_p(v_{NT}/T)$, using $r_{NT}^2=n^2/v_{NT}$ and $n=NT$, and the second is $O_p(v_{NT}^{1/2}(v_{NT}/T)^{1/2})=O_p(v_{NT}/\sqrt{T})$. The same argument gives a change of order $O_p(NT^2r_{NT}^{-2})=O_p(v_{NT}/N)$ for the sum clustered by unit and $O_p(v_{NT}/(NT))$ for the heteroskedasticity-robust sum, and the corresponding cross terms are of order $O_p(v_{NT}/\sqrt{N})$ and smaller. All these terms are $o_p(v_{NT})$, so $v_{NT}^{-1}Q(\hat\psi)\xrightarrow{p}\Omega$. Setting negative eigenvalues to zero does not affect the limit, which is positive definite. With the Jacobian estimated at $\hat\theta^J$, which is consistent, $r_{NT}^2\hat V\xrightarrow{p}S\Omega S'$.

Step 4 (the AR statistic). By Lemma~\ref{lem:IRF}, $\lambda_{h,s}=e_s'\Phi^h\Gamma/e_1'\Gamma$, so under $H_0:\lambda_{h,s}=\lambda_0$, $a(\lambda_0)'\Gamma=0$. Expanding,
\[
r_{NT}\,\hat a(\lambda_0)'\hat\Gamma=r_{NT}\,a(\lambda_0)'(\hat\Gamma-\Gamma)+r_{NT}\,e_s'(\hat\Phi^h-\Phi^h)\Gamma+r_{NT}\,e_s'(\hat\Phi^h-\Phi^h)(\hat\Gamma-\Gamma).
\]
The third term is $O_p(r_{NT}^{-1})$ by Step 3. The second equals $r_{NT}(\Gamma'\otimes e_s')G_h\,\mathrm{vec}(\hat\Phi-\Phi)+o_p(1)$, because $\Phi\mapsto\Phi^h$ is a polynomial. Hence $r_{NT}\,\hat a(\lambda_0)'\hat\Gamma=d'r_{NT}(\hat\theta^J-\theta)+o_p(1)$ with $d=\bigl((\Gamma'\otimes e_s')G_h,\ a(\lambda_0)'\bigr)'$, and it converges to a normal variable with variance $\sigma^2=d'S\Omega S'd$, bounded away from zero by assumption. The estimator $r_{NT}^2\hat\sigma^2_{h,s}(\lambda_0)=\hat d'\,r_{NT}^2\hat V\hat d$ converges in probability to $\sigma^2$ by Step 3 and the continuity of $d$. The normalization $r_{NT}$ cancels in the ratio, so $\text{AR}_{h,s}(\lambda_0)\xrightarrow{d}\chi^2_1$ under the null, and
\[
\Pr\bigl(\lambda_{h,s}\in\text{CS}^{\text{AR}}_{h,s}(1-\alpha)\bigr)=\Pr\bigl(\text{AR}_{h,s}(\lambda_{h,s})\leq\chi^2_{1,1-\alpha}\bigr)\to1-\alpha.
\]
No step uses the size of $\Gamma_1$: the statistic is linear in $\hat\Gamma$ and never divides by $\hat\Gamma_1$, the variance $\sigma^2$ stays bounded away from zero because $a(\lambda_0)\neq0$ and the block of $S\Omega S'$ for $\Gamma$ is positive definite, and the term in $\hat\Phi$ vanishes if $\Gamma$ drifts to zero. For the cumulated response, replace $e_s'\Phi^h$ by $\sum_{k=0}^{h}e_s'\Phi^k$ and $G_h$ by $\sum_{k=0}^{h}G_k$. For the cumulative multiplier, replace $a(\lambda_0)$ by $c_H(M_0)$ and $G_h$ by $\sum_{k=0}^{H}G_k$; the corresponding variance is bounded away from zero by the same argument, because $c_H(M_0)\neq0$.

Step 5 (the uncorrected estimator). By Step 1, $r_{NT}(\hat\theta-\theta)=S\,r_{NT}\bar\psi+r_{NT}B/T+o_p(1)$. If $v_{NT}/(N^2T)$ is bounded away from zero, $r_{NT}$ is of order $\sqrt{T}$ and $r_{NT}B/T\to0$, so Steps 3 and 4 apply to $\hat\theta$ without change.
\end{proof}
\section{Several instruments and the conditional likelihood ratio test\label{sec:AppCLR}}

The main text uses one instrument, for which the Anderson-Rubin (AR) set is optimal. This appendix states the construction and the size result when $Z_{it}$ is a $k\times1$ vector of instruments, $k\geq1$. The AR set remains valid but wastes degrees of freedom when some instruments are weak or irrelevant, and the conditional likelihood ratio (CLR) test addresses this \citep{moreira2003conditional,andrews2006optimal}. The test below conditions on the part of the first stage that is asymptotically independent of the null moments, as in \citet{kleibergen2005testing} and \citet{andrews2016conditional}, so that it can be built from the delta-method covariance of the VAR coefficients and of the instrument moments, with any dependence structure that this covariance can capture.

Let $\ddot{Z}_{it}$ be the $k\times1$ vector of instruments after removing the unit and period effects, $\Gamma=\E[\tilde{x}_{it}\ddot{Z}_{it}']$ the $m\times k$ matrix with columns $\Gamma_{\cdot l}$, and $\hat\Gamma=n^{-1}\sum_{i,t}\hat{\tilde{x}}_{it}\ddot{Z}_{it}'$ its estimator, computed with the residuals of the corrected estimator $\hat\Phi^J$ of Section~\ref{sec:Inference}. Under Assumption~\ref{assu:Shock} applied to each instrument, the argument of Lemma~\ref{lem:IRF}(a) gives $\Gamma=b\,\kappa'$ with $\kappa=\E[\ddot{Z}_{it}\eta_{it}]\in\mathbb{R}^k$, so every column of $\Gamma$ is proportional to $b$ and the first row of $\Gamma$ is $\gamma'\equiv e_1'\Gamma=\kappa'$. The restriction that $\Gamma$ has rank one is an overidentifying restriction of the model, which I do not test here. For $\lambda_{h,s}=e_s'\Phi^hb$, put $a(\lambda_0)'=e_s'\Phi^h-\lambda_0e_1'$, so that
\[
g(\lambda_0)\equiv\Gamma'a(\lambda_0)=\bigl(\lambda_{h,s}-\lambda_0\bigr)\,\kappa ,
\]
using $e_1'b=1$. Hence $g(\lambda_0)=0$ if and only if $\lambda_{h,s}=\lambda_0$ whenever $\kappa\neq0$, and $g(\lambda_0)=0$ for every $\lambda_0$ when $\kappa=0$, in which case the set is the whole line. Let $\theta=(\text{vec}(\Phi)',\text{vec}(\Gamma)')'$ with dimension $d$, and write $g_l(\lambda_0)=f_l(\theta)-\lambda_0\gamma_l(\theta)$ with $f_l=e_s'\Phi^h\Gamma_{\cdot l}$ and $\gamma_l=e_1'\Gamma_{\cdot l}$. Let $J_f$ and $J_\gamma$ be the $k\times d$ Jacobians of $f=(f_1,\dots,f_k)'$ and $\gamma$ with respect to $\theta$. Row $l$ of $J_f$ has the block $(\Gamma_{\cdot l}'\otimes e_s')\,\partial\,\text{vec}(\Phi^h)/\partial\,\text{vec}(\Phi)'$ for $\text{vec}(\Phi)$ and the block $e_s'\Phi^h$ in the position of $\Gamma_{\cdot l}$, and row $l$ of $J_\gamma$ has $e_1'$ in the position of $\Gamma_{\cdot l}$ and zeros elsewhere. For the cumulated response, replace $e_s'\Phi^h$ by $\sum_{q=0}^{h}e_s'\Phi^q$ in $a(\lambda_0)$, $f$ and $J_f$. For the cumulative multiplier $M_{j,H}$, replace $a(\lambda_0)'$ by $c_H(M_0)'$, which changes $f$ and $\gamma$ accordingly.

\begin{lem}[Asymptotic normality with $k$ instruments]\label{assu:HighLevelK}
Suppose Assumption~\ref{assu:Regularity} holds with the moment contributions $\psi_{it}=\bigl(\mathrm{vec}(X_{it}\tilde{x}_{it}')',\ \mathrm{vec}(\tilde{x}_{it}\ddot{Z}_{it}'-\Gamma)'\bigr)'$ built from the $k$ instruments, along sequences of parameters $\theta_n$ in a compact set that allow $\Gamma_n\to0$. Let $\hat\theta_n=(\mathrm{vec}(\hat\Phi^J)',\mathrm{vec}(\hat\Gamma)')'$, let $\hat V_n$ be the two-way clustered estimator of its covariance, and let $r_{NT}$ be the rate of Assumption~\ref{assu:Regularity}(iii). Then $r_{NT}(\hat\theta_n-\theta_n)\xrightarrow{d}N(0,W)$ and $r_{NT}^2\hat V_n\xrightarrow{p}W$, with $W$ positive definite.
\end{lem}

\begin{proof}
Steps 1 to 3 of the proof of Proposition~\ref{prop:AR-coverage} use the number of instruments only through the dimension of the moment contributions, and they do not use the size of $\Gamma_n$. They therefore apply unchanged to the contributions defined above, and give the two limits with $W=S\Omega S'$, which is positive definite because $\Omega$ is positive definite by Assumption~\ref{assu:Regularity}(iii) and the Jacobian $S$ has full rank by Assumption~\ref{assu:Regularity}(i).
\end{proof}

For a null value $\lambda_0$, let $J_g(\lambda_0)=J_f-\lambda_0J_\gamma$ and
\[
\Sigma(\lambda_0)=J_g(\lambda_0)\,W\,J_g(\lambda_0)',\qquad
\Sigma_{\gamma g}(\lambda_0)=J_\gamma\,W\,J_g(\lambda_0)',\qquad
\Omega(\lambda_0)=J_\gamma WJ_\gamma'-\Sigma_{\gamma g}\Sigma^{-1}\Sigma_{\gamma g}',
\]
which are the asymptotic covariance of the normalized null moments $r_{NT}\hat g$, the covariance of $r_{NT}\hat\gamma$ with $r_{NT}\hat g$, and the covariance of $r_{NT}\hat\gamma$ after projecting out $r_{NT}\hat g$. Their estimators $\hat\Sigma,\hat\Sigma_{\gamma g},\hat\Omega$ replace $W$ by $\hat V_n$ and the Jacobians by their values at $\hat\theta_n$, so that they estimate the covariances of $\hat g$ and $\hat\gamma$ themselves, and $r_{NT}^2\hat\Sigma$, $r_{NT}^2\hat\Sigma_{\gamma g}$ and $r_{NT}^2\hat\Omega$ are consistent for $\Sigma$, $\Sigma_{\gamma g}$ and $\Omega$ by Lemma~\ref{assu:HighLevelK}. With $\hat g(\lambda_0)=\hat\Gamma'\hat a(\lambda_0)$, define
\begin{align*}
\text{AR}_k(\lambda_0)&=\hat g'\hat\Sigma^{-1}\hat g, &
\tilde D&=-\hat\gamma+\hat\Sigma_{\gamma g}\hat\Sigma^{-1}\hat g,\\
K(\lambda_0)&=\frac{(\tilde D'\hat\Sigma^{-1}\hat g)^2}{\tilde D'\hat\Sigma^{-1}\tilde D}, &
\hat r(\lambda_0)&=\tilde D'\hat\Omega^{-1}\tilde D,
\end{align*}
where $K=0$ if $\tilde D'\hat\Sigma^{-1}\tilde D=0$. The three statistics are invariant to a common rescaling of $\hat V_n$, so they equal the statistics built from the normalized quantities $r_{NT}\hat g$, $r_{NT}\tilde D$ and $r_{NT}^2\hat V_n$, and they can be computed without knowing the rate $r_{NT}$. Finally,
\[
\text{CLR}(\lambda_0)=\tfrac12\Bigl[\text{AR}_k-\hat r+\sqrt{(\text{AR}_k+\hat r)^2-4\hat r\,(\text{AR}_k-K)}\Bigr].
\]
The vector $\tilde D$ is the estimated derivative of the moments with respect to $\lambda_0$, $-\hat\gamma$, purged of its correlation with $\hat g$, and $\hat r$ measures its size relative to its own covariance. Let $c_{1-\alpha}(r,k)$ be the $(1-\alpha)$ quantile of
\[
\tfrac12\Bigl[Q-r+\sqrt{(Q+r)^2-4r\,(Q-z^2)}\Bigr],\qquad Q=z^2+\chi^2_{k-1},
\]
with $z\sim N(0,1)$ independent of $\chi^2_{k-1}$ (and $\chi^2_0\equiv0$). It depends only on $(r,k)$. Lemma~\ref{lem:crit} shows that, since $\text{AR}_k-K=\chi^2_{k-1}$ in this display, the statistic can be written $\tfrac12[\text{AR}_k-\hat r+\sqrt{(\text{AR}_k-\hat r)^2+4\hat rK}]$, and that $c_{1-\alpha}(r,k)$ is continuous and decreasing in $r$, from $\chi^2_{k,1-\alpha}$ at $r=0$ to $\chi^2_{1,1-\alpha}$ as $r\to\infty$. The sets are
\[
\text{CS}^{\text{CLR}}_{h,s}(1-\alpha)=\{\lambda_0:\ \text{CLR}(\lambda_0)\leq c_{1-\alpha}(\hat r(\lambda_0),k)\},
\]
and $\text{CS}^{\text{AR}_k}$ and $\text{CS}^{K}$ are defined with $\chi^2_{k,1-\alpha}$ and $\chi^2_{1,1-\alpha}$. Because $a(\lambda_0)$ and therefore $\hat\Sigma$ and $\hat r$ depend on $\lambda_0$ in a non-linear way, the sets are computed by inversion over a grid of $\lambda_0$ and need not be intervals; they can be unions of intervals, a complement of an interval or the whole line.

\begin{lem}[Critical value of the CLR test]\label{lem:crit}
Let $z\sim N(0,1)$ and $X\sim\chi^2_{k-1}$ be independent ($X\equiv0$ if $k=1$), $Q=z^2+X$, and for $r\geq0$
\[
L(r)=\tfrac12\Bigl[Q-r+\sqrt{(Q+r)^2-4rX}\Bigr].
\]
Then:
\begin{enumerate}
\item[(i)] $L(r)=\tfrac12\bigl[Q-r+\sqrt{(Q-r)^2+4rz^2}\bigr]$, so the argument of the square root is nonnegative;
\item[(ii)] $L(0)=Q$, $L(r)\to z^2$ as $r\to\infty$, $r\mapsto L(r)$ is nonincreasing, and strictly decreasing when $z^2X>0$; hence $z^2\leq L(r)\leq Q$ for all $r$;
\item[(iii)] for $\alpha\in(0,1)$, the $(1-\alpha)$ quantile $c_{1-\alpha}(r,k)$ of $L(r)$ is continuous and nonincreasing in $r\in[0,\infty)$, strictly decreasing if $k\geq2$, with $c_{1-\alpha}(0,k)=\chi^2_{k,1-\alpha}$, $\lim_{r\to\infty}c_{1-\alpha}(r,k)=\chi^2_{1,1-\alpha}$ and $\chi^2_{1,1-\alpha}\leq c_{1-\alpha}(r,k)\leq\chi^2_{k,1-\alpha}$. For $k=1$, $c_{1-\alpha}(r,1)=\chi^2_{1,1-\alpha}$ for all $r$.
\end{enumerate}
\end{lem}

\begin{proof}
(i) $(Q+r)^2-4rX=(Q-r)^2+4r(Q-X)=(Q-r)^2+4rz^2$.

(ii) Let $D(r)=\sqrt{(Q-r)^2+4rz^2}$. Then $L(0)=\tfrac12(Q+Q)=Q$. For $D(r)>0$,
\[
L'(r)=\tfrac12\Bigl[-1+\frac{2z^2-Q+r}{D(r)}\Bigr].
\]
Expanding the squares, $D(r)^2-(2z^2-Q+r)^2=4rz^2-4z^4+4z^2(Q-r)=4z^2(Q-z^2)=4z^2X\geq0$, so $|2z^2-Q+r|\leq D(r)$ and $L'(r)\leq0$, with strict inequality unless $2z^2-Q+r=D(r)$, which requires $z^2X=0$. The set where $D(r)=0$ has probability zero. For large $r$, $D(r)=(r-Q)+2rz^2/(r-Q)+O(r^{-1})$, so $L(r)=rz^2/(r-Q)+o(1)\to z^2$. A nonincreasing function with limit $z^2$ at infinity and value $Q$ at zero lies between them.

(iii) If $k=1$ then $X=0$, $Q=z^2$ and $L(r)=z^2$ for all $r$. Let $k\geq2$. For fixed $r>0$ and $z\neq0$, $\partial L/\partial Q=\tfrac12[1+(Q-r)/D(r)]>0$ because $D(r)>|Q-r|$, so $L(r)$ is strictly increasing in $X$. Since $X$ has a positive density on $(0,\infty)$ given $z$, $L(r)$ has a distribution with a continuous, strictly increasing cdf on $(0,\infty)$ (for $r=0$ it is $\chi^2_k$), so its quantile is unique. If $r^{(m)}\to r$, $L(r^{(m)})\to L(r)$ almost surely by continuity of $L$ in $r$, hence in distribution, and quantiles converge when the limit law has a unique quantile; this gives continuity of $c_{1-\alpha}(\cdot,k)$. By (ii), $L(r')\leq L(r)$ pointwise for $r'>r$, with strict inequality with probability one, so $c_{1-\alpha}(r',k)<c_{1-\alpha}(r,k)$. Since $L(r)\downarrow z^2$ almost surely as $r\to\infty$, the quantiles converge to those of $z^2\sim\chi^2_1$, and $L(0)=Q\sim\chi^2_k$ gives the value at zero. The bounds follow from $z^2\leq L(r)\leq Q$ pointwise.
\end{proof}

\begin{thm}[Size of the AR, $K$ and CLR tests with $k$ instruments]\label{thm:CLR}
Suppose Assumptions~\ref{assu:Linearity} to~\ref{assu:Regularity} hold, with Assumption~\ref{assu:Shock} applied to each instrument and Assumption~\ref{assu:Regularity} applied to the moment contributions of Lemma~\ref{assu:HighLevelK}, and that $\lambda_0=\lambda_{h,s}$ is the true value, with $a(\lambda_0)\neq0$. Then, as $N,T\to\infty$ and for any sequence of parameters as in Lemma~\ref{assu:HighLevelK}, including $\Gamma_n\to0$:
\begin{enumerate}
\item[(a)] $\text{AR}_k(\lambda_0)\xrightarrow{d}\chi^2_k$ and $K(\lambda_0)\xrightarrow{d}\chi^2_1$;
\item[(b)] conditionally on $\tilde D$, the limit of $(\text{AR}_k,K)$ is $(z^2+\chi^2_{k-1},\,z^2)$ with $z\sim N(0,1)$ independent of $\chi^2_{k-1}$, so that
\[
\lim_{n\to\infty}\Pr\bigl(\text{CLR}(\lambda_0)>c_{1-\alpha}(\hat r(\lambda_0),k)\bigr)=\alpha;
\]
\item[(c)] consequently $\text{CS}^{\text{CLR}}_{h,s}$, $\text{CS}^{\text{AR}_k}_{h,s}$ and $\text{CS}^{K}_{h,s}$ each contain $\lambda_{h,s}$ with limiting probability $1-\alpha$, and the same holds for the cumulated response and for the cumulative multiplier $M_{j,H}$;
\item[(d)] for $k=1$, $K=\text{AR}_1=\text{CLR}$ and the three sets equal the AR set of Section~\ref{sec:Inference}.
\end{enumerate}
\end{thm}

The theorem is a statement of size, and it holds for any positive definite $W$, including the two-way clustered covariance of Section~\ref{sec:Inference}, because the conditioning argument uses only the joint normality of $(\hat g,\hat\gamma)$ and the consistency of the covariance estimator. It makes no claim of optimality. The test is the conditional quasi-likelihood-ratio (CQLR) test of \citet{kleibergen2005testing}: in the notation of \citet[eq.~(2.7)]{moreira2019optimal}, $\text{QLR}=\tfrac12\bigl[\text{AR}-r(T)+\sqrt{(\text{AR}-r(T))^2+4\,\text{LM}\cdot r(T)}\bigr]$ with $r(T)=T'T$, and the critical value is the conditional quantile given $T$ as in \citet{moreira2003conditional}. Here $S=\hat\Sigma^{-1/2}\hat g$ and $T=\hat\Omega^{-1/2}\tilde D$ correspond to their $S$ and $T$, $K$ to their LM statistic and $\hat r$ to $r(T)$. When $\Omega(\lambda_0)=c\,\Sigma(\lambda_0)$, the analogue of the Kronecker structure $\Sigma=\Omega\otimes\Phi$ of \citet[Section~2.1]{moreira2019optimal}, which contains the homoskedastic IV model, $S$ and $T$ are the statistics of \citet{moreira2003conditional} and \citet{andrews2006optimal}, $\hat r=T'T=Q_T$, $K=(S'T)^2/T'T$, and the test is Moreira's CLR test, with the power properties established there. With $k=1$ the AR, $K$ and CQLR tests coincide \citep[Proposition~3]{moreira2019optimal}. For heteroskedastic and autocorrelated errors, \citet{moreira2019optimal} show that the AR, score and (pseudo) likelihood ratio tests of \citet{kleibergen2005testing} belong to the class of strongly unbiased tests, and propose weighted-average-power tests (MM1-SU and MM2-SU) that are asymptotically equivalent to the likelihood ratio test under strong instruments. In their simulations these tests generally have higher power than the CQLR test under weak instruments, with small losses where they do not. The result establishes validity, and the comparison of power with these alternatives is left to their simulations.

The moments $g(\lambda_0)$ are exactly zero at the true value for every $n$, whatever the size of $\Gamma_n$, so the limit of $r_{NT}\hat g$ does not depend on instrument strength, and neither does the conditional distribution in part (b). Strength enters only through $\hat r$, which is small when the instruments are weak, in which case $c_{1-\alpha}(\hat r,k)$ approaches $\chi^2_{k,1-\alpha}$ and CLR approaches the $\text{AR}_k$ test, and large when they are strong, in which case $c_{1-\alpha}(\hat r,k)\to\chi^2_{1,1-\alpha}$ and CLR approaches the $K$ test. The gain over $\text{AR}_k$ arises when the instruments are strong in a few directions, because the AR critical value is calibrated to $k$ degrees of freedom.

\begin{proof}[Proof of Theorem~\ref{thm:CLR}]
Write $\theta_n$ for the true parameter. By the definition of $g$ and $\lambda_0=\lambda_{h,s}$, $g_n(\lambda_0)=0$ for every $n$. The map $\theta\mapsto g(\theta;\lambda_0)$ is a polynomial in $\theta$, so a second-order expansion around $\theta_n$ gives
\[
r_{NT}\,\hat g(\lambda_0)=J_{g,n}(\lambda_0)\,r_{NT}(\hat\theta_n-\theta_n)+O_p(r_{NT}^{-1}),
\]
where the remainder is of order $r_{NT}\|\hat\theta_n-\theta_n\|^2$, uniformly in the compact set, whatever the value of $\Gamma_n$. Likewise $r_{NT}(\hat\gamma-\gamma_n)=J_{\gamma,n}\,r_{NT}(\hat\theta_n-\theta_n)$ exactly. By Lemma~\ref{assu:HighLevelK} and the continuity of the Jacobians in $\theta$, $(r_{NT}\hat g,\ r_{NT}(\hat\gamma-\gamma_n))$ is asymptotically jointly normal with mean zero and covariance given by $\Sigma$, $\Sigma_{\gamma g}$ and $J_\gamma WJ_\gamma'$, and $r_{NT}^2\hat\Sigma$, $r_{NT}^2\hat\Sigma_{\gamma g}$ and $r_{NT}^2\hat\Omega$ are consistent for the population matrices. Since $a(\lambda_0)\neq0$, the block of $J_g$ that multiplies $\text{vec}(\Gamma)$ is $I_k\otimes a(\lambda_0)'$, which has rank $k$, so $\Sigma$ is nonsingular because $W$ is positive definite; similarly $\Omega$ is nonsingular.

(a) $\text{AR}_k=(r_{NT}\hat g)'(r_{NT}^2\hat\Sigma)^{-1}(r_{NT}\hat g)\xrightarrow{d}\chi^2_k$ by the continuous mapping theorem.

(b) Let $S_n=\hat\Sigma^{-1/2}\hat g=(r_{NT}^2\hat\Sigma)^{-1/2}r_{NT}\hat g$, $D_n=\hat\Sigma^{-1/2}\tilde D=(r_{NT}^2\hat\Sigma)^{-1/2}r_{NT}\tilde D$ and $u_n=D_n/\|D_n\|$ (any unit vector if $D_n=0$), so that $\text{AR}_k=\|S_n\|^2$, $K=(S_n'u_n)^2$ and $\hat r=D_n'M_nD_n$ with $M_n=\hat\Sigma^{1/2}\hat\Omega^{-1}\hat\Sigma^{1/2}\xrightarrow{p}M=\Sigma^{1/2}\Omega^{-1}\Sigma^{1/2}$ positive definite. By construction $\tilde D$ is uncorrelated with $\hat g$ in the limit, so $S_n\xrightarrow{d}N(0,I_k)$ and $D_n-E_n$, with $E_n=-r_{NT}\Sigma^{-1/2}\gamma_n$ nonrandom, are asymptotically jointly normal and independent of each other. It is enough to show that every subsequence has a further subsequence along which $\Pr(\text{CLR}>c_{1-\alpha}(\hat r,k))\to\alpha$. Take a further subsequence along which either (A) $E_n\to E$ is finite or (B) $\|E_n\|\to\infty$.

Case (A). Then $(S_n,D_n)\xrightarrow{d}(S,D)$ with $S\sim N(0,I_k)$ and $D=E+\xi$, where $\xi$ is a nondegenerate normal vector independent of $S$. The rejection event is a function of $(S_n,D_n)$ that is continuous except on a null set (note $D\neq0$ almost surely and, by Lemma~\ref{lem:crit}(iii), $c_{1-\alpha}(r,k)$ is continuous in $r$ and the law of the statistic is continuous), so $\Pr(\text{CLR}>c_{1-\alpha}(\hat r,k))\to\Pr(\text{CLR}(S,D)>c_{1-\alpha}(r(D),k))$ with $r(D)=D'MD$. Conditionally on $D$, $u=D/\|D\|$ is fixed and $S\sim N(0,I_k)$, so $K=(S'u)^2\sim\chi^2_1$ and $\text{AR}_k-K=\|S\|^2-(S'u)^2\sim\chi^2_{k-1}$ are independent. The statistic therefore has the law of $L(r(D))$ in Lemma~\ref{lem:crit}, whose $(1-\alpha)$ quantile is $c_{1-\alpha}(r(D),k)$ and whose law is continuous. The conditional rejection probability equals $\alpha$ for every $D$, hence so does the unconditional probability.

Case (B). Here $\|D_n\|\to\infty$ in probability, so $\hat r\geq\lambda_{\min}(M_n)\|D_n\|^2\to\infty$ in probability. Since $r_{NT}\tilde D=-r_{NT}\gamma_n+O_p(1)$ and $r_{NT}^2\hat\Sigma\xrightarrow{p}\Sigma$, $D_n=\Sigma^{-1/2}\Sigma^{1/2}E_n(1+o_p(1))$, so $u_n-\bar u_n\xrightarrow{p}0$ with $\bar u_n=E_n/\|E_n\|$ nonrandom. Hence $K=(S_n'\bar u_n)^2+o_p(1)\xrightarrow{d}\chi^2_1$. By Lemma~\ref{lem:crit}(ii), $L$ is monotone in $r$ and continuous in $(Q,z^2,r)$ with limit $z^2$ as $r\to\infty$, which is continuous in $(Q,z^2)$, so by Dini's theorem $L-z^2\to0$ uniformly on compact sets of $(Q,z^2)$ as $r\to\infty$. Since $(\text{AR}_k,K)$ is tight and $\hat r\to\infty$ in probability, $\text{CLR}-K\xrightarrow{p}0$. By Lemma~\ref{lem:crit}(iii), $c_{1-\alpha}(\hat r,k)\xrightarrow{p}\chi^2_{1,1-\alpha}$. Therefore $\Pr(\text{CLR}>c_{1-\alpha}(\hat r,k))=\Pr(K>\chi^2_{1,1-\alpha}+o_p(1))\to\Pr(\chi^2_1>\chi^2_{1,1-\alpha})=\alpha$.

(c) Follows from (a) and (b) because $\lambda_{h,s}\in\text{CS}$ if and only if the test of $H_0:\lambda_{h,s}=\lambda_0$ at $\lambda_0=\lambda_{h,s}$ does not reject. For the cumulated response replace $e_s'\Phi^h$ by $\sum_{q=0}^{h}e_s'\Phi^q$, and for $M_{j,H}$ replace $a(\lambda_0)$ by $c_H(M_0)$; in both cases the moments are linear in $\Gamma$ and polynomial in $\Phi$, which is all the argument uses.

(d) For $k=1$, $\tilde D=-\hat\gamma+\hat\Sigma_{\gamma g}\hat\Sigma^{-1}\hat g$ is a scalar, $K=\hat g^2/\hat\Sigma=\text{AR}_1$, so $\text{AR}_k-K=0$ and CLR$=\tfrac12[\text{AR}_1-\hat r+|\text{AR}_1+\hat r|]=\text{AR}_1$, and $c_{1-\alpha}(r,1)=\chi^2_{1,1-\alpha}$.
\end{proof}

\section{Additional results\label{sec:AppResults}}

Table~\ref{tab:cns_diag} reports, for each specification of the application, the first-stage statistics, the residual autocorrelation and the largest modulus of the eigenvalues of the companion matrix of the estimated PVAR. Table~\ref{tab:cns_irf} reports the response of nontradable payroll at selected horizons under the alternative specifications shown in the robustness figure of the paper. Table~\ref{tab:mc_gmm} compares, in the calibrated design of the Monte Carlo simulations, the within estimator with and without the half-panel jackknife correction and GMM after forward orthogonal deviations with lagged levels as instruments. Table~\ref{tab:mc_county} reports, for every design, the coverage of the Anderson-Rubin sets with the covariance clustered by county alone, by quarter, and in both dimensions.
Table~\ref{tab:mc_persistent} reports the calibrated design with a more persistent aggregate shift, $\rho_g=0.67$ and $\rho_g=0.9$. The instrument is then correlated with past policy shocks, which Assumption~\ref{assu:Shock}(ii) allows.

The parameters of the simulations of Section~6 are calibrated to the application, estimated without population weights because the simulated counties have equal weight. The matrices $\Phi_l$ and the effect $\beta$ are the estimates of the PVAR-IV. The exposures $s_i$ are the county averages of lagged stock wealth relative to labor income, and $\rho_g$ and $\sigma_e^2$ are the autoregressive coefficient and innovation variance of the quarterly market return. The loading $\gamma$ is obtained from the first-stage slope of the policy residual on the instrument. The loading $\psi_Y$ and the variance of $v_{it}$ reproduce the variance of the outcome residual net of the policy effect and the ratio between the quarter-clustered and the heteroskedasticity-robust variance of the moment entering the AR statistic. The loading $\psi_W$ is set so that the median two-way clustered first-stage statistic of the calibrated design equals the one in the application, and the variance of $u_{it}$ reproduces the variance of the policy residual. The correlation between $u_{it}$ and $v_{it}$ is the correlation of the corresponding residuals, and the county effects are normal with the dispersion of the county means in the data.

\begin{table}[ht]\centering\footnotesize
\resizebox{\textwidth}{!}{\begin{tabular}{lrrrrrrr}\toprule
Specification & $p$ & $n$ & $F$, two-way & $F$, quarter & max $|\hat\rho|$ & max $|t|$ & max $|\mu|$\\\midrule
Baseline & 8 & 240,121 & 60.3 & 59.4 & 0.200 & 41.7 & 0.848\\
Baseline & 12 & 228,541 & 37.2 & 36.6 & 0.204 & 37.1 & 0.937\\
Baseline & 16 (BIC) & 216,962 & 36.2 & 35.6 & 0.256 & 40.4 & 0.997\\
Unweighted & 8 & 240,121 & 82.8 & 83.1 & 0.190 & 13.6 & 0.848\\
Fixed-dividend instrument & 8 & 259,931 & 132.4 & 130.8 & 0.178 & 40.2 & 0.853\\
Within estimator, no correction & 8 & 240,121 & 64.1 & 63.2 & 0.125 & 27.8 & 0.851\\
FOD-GMM & 8 & 237,226 & 82.4 & 81.2 & 0.111 & 29.0 & 0.850\\
\bottomrule\end{tabular}}
\caption{First stage and residual diagnostics. $F$ is the robust first-stage statistic, two-way clustered by county and quarter or clustered by quarter. $\hat\rho$ is the autocorrelation of the VAR residuals at lags 1 to 4, largest absolute value over equations and lags, with its county-clustered $t$-statistic. $|\mu|$: largest modulus of the eigenvalues of the companion matrix of the estimated PVAR.\label{tab:cns_diag}}
\end{table}

\begin{table}[ht]\centering\footnotesize
\resizebox{\textwidth}{!}{\begin{tabular}{llccccc}\toprule
Specification & $p$ & $h=0$ & $h=4$ & $h=7$ & $h=12$ & $h=16$\\\midrule
Baseline & 8 & 0.53 & 2.83 & 2.95 & 2.56 & 2.55\\
 & & [-0.01, 1.20] & [1.55, 4.70] & [1.29, 5.51] & [0.46, 5.61] & [0.43, 5.53]\\[2pt]
Baseline & 12 & 0.97 & 3.12 & 3.20 & 3.68 & 3.97\\
 & & [0.31, 2.09] & [1.70, 5.52] & [1.07, 6.94] & [-0.03, 9.90] & [-1.18, 12.51]\\[2pt]
Baseline & 16 (BIC) & 1.17 & 3.20 & 3.60 & 4.87 & 5.63\\
 & & [0.49, 2.52] & [1.63, 5.85] & [1.23, 7.83] & [0.58, 12.41] & [-0.91, 16.82]\\[2pt]
Unweighted & 8 & 0.42 & 2.89 & 2.59 & 1.93 & 2.05\\
 & & [-0.32, 1.22] & [0.92, 5.41] & [0.57, 5.46] & [-0.03, 4.63] & [0.35, 4.37]\\[2pt]
Fixed-dividend instrument & 8 & 0.70 & 2.68 & 2.45 & 1.82 & 1.87\\
 & & [0.13, 1.43] & [1.28, 4.50] & [0.81, 4.68] & [0.09, 4.08] & [0.29, 3.91]\\[2pt]
Within estimator, no correction & 8 & 0.55 & 2.84 & 2.83 & 2.30 & 2.33\\
 & & [0.02, 1.17] & [1.59, 4.67] & [1.28, 5.21] & [0.54, 4.86] & [0.70, 4.68]\\[2pt]
FOD-GMM & 8 & 0.42 & 2.41 & 2.21 & 1.76 & 1.79\\
 & & [-0.15, 1.05] & [1.06, 4.06] & [0.46, 4.31] & [-0.22, 4.01] & [-0.04, 3.90]\\[2pt]
\midrule Local projection IV & & 0.66 & 2.64 & 2.61 & 2.86 & 2.33\\
 & & (0.28) & (0.70) & (0.95) & (1.09) & (1.10)\\
\bottomrule\end{tabular}}
\caption{Response of nontradable payroll, in basis points, to a stock-wealth shock equal to 1\% of annual labor income. PVAR-IV rows: point estimate and 95\% Anderson-Rubin set, two-way clustered by county and quarter; within estimator with half-panel jackknife correction unless the row states otherwise. Local projection IV: the estimator of the original study, with two-way clustered standard errors in parentheses.\label{tab:cns_irf}}
\end{table}

\begin{table}[ht]\centering\footnotesize
\begin{tabular}{r rrr rr rr}\toprule
 & \multicolumn{3}{c}{AR coverage, two-way (\%)} & \multicolumn{2}{c}{Median length} & \multicolumn{2}{c}{Median bias}\\
\cmidrule(lr){2-4}\cmidrule(lr){5-6}\cmidrule(lr){7-8}
$h$ & Within & Jackknife & FOD-GMM & Jackknife & FOD-GMM & Jackknife & FOD-GMM\\\midrule
0 & 95.4 & 95.4 & 95.4 & 2.25 & 2.33 & 0.01 & -0.00\\
4 & 94.6 & 93.8 & 95.8 & 2.77 & 3.18 & 0.08 & 0.03\\
7 & 92.4 & 92.0 & 96.8 & 3.67 & 4.35 & 0.09 & 0.01\\
12 & 91.2 & 92.2 & 96.4 & 3.62 & 4.28 & 0.05 & -0.01\\
\bottomrule\end{tabular}
\caption{Calibrated design of Section 6: the within estimator without and with the half-panel jackknife correction, and GMM after forward orthogonal deviations with lagged levels as instruments. Coverage of nominal 95\% Anderson-Rubin sets clustered by county and quarter, median length of the sets, and median bias of the point estimate. 500 replications.\label{tab:mc_gmm}}
\end{table}

\begin{table}[ht]\centering\footnotesize
\begin{tabular}{llr rrr}\toprule
 & & & \multicolumn{3}{c}{AR coverage, within estimator (\%)}\\ \cmidrule(lr){4-6}
Design & $h$ & $F$ & county & quarter & two-way\\\midrule
Calibrated & 0 & 78.1 & 72.2 & 95.2 & 95.4\\
 & 4 &  & 71.4 & 94.0 & 94.6\\
 & 7 &  & 68.8 & 91.6 & 92.4\\
 & 12 &  & 62.0 & 90.0 & 91.2\\
\addlinespace
Weak, $c=0.4$ & 0 & 13.6 & 4.6 & 94.0 & 94.2\\
 & 4 &  & 4.8 & 92.8 & 93.4\\
 & 7 &  & 8.2 & 93.8 & 94.0\\
 & 12 &  & 6.2 & 94.8 & 95.2\\
\addlinespace
Weak, $c=0.2$ & 0 & 3.6 & 4.0 & 92.6 & 93.0\\
 & 4 &  & 5.4 & 94.8 & 95.4\\
 & 7 &  & 5.8 & 94.8 & 95.0\\
 & 12 &  & 6.6 & 94.2 & 94.6\\
\addlinespace
Persistent, $\rho_g=0.67$ & 0 & 35.2 & 71.8 & 92.4 & 92.8\\
 & 4 &  & 71.2 & 93.2 & 94.0\\
 & 7 &  & 71.2 & 93.2 & 93.6\\
 & 12 &  & 64.8 & 92.6 & 93.0\\
\addlinespace
Persistent, $\rho_g=0.9$ & 0 & 10.8 & 70.0 & 88.8 & 89.8\\
 & 4 &  & 69.4 & 92.0 & 92.4\\
 & 7 &  & 63.8 & 91.8 & 92.4\\
 & 12 &  & 62.8 & 91.6 & 92.2\\
\addlinespace
\bottomrule\end{tabular}
\caption{Coverage of nominal 95\% Anderson-Rubin sets for the cumulated response, within estimator, with the covariance clustered by county alone, by quarter, and in both dimensions, for the designs of Section 6. $F$: median first-stage statistic, two-way clustered. 500 replications per design.\label{tab:mc_county}}
\end{table}

\begin{table}[ht]\centering\footnotesize\setlength{\tabcolsep}{4pt}
\resizebox{\textwidth}{!}{\begin{tabular}{llr rrr rr rr}\toprule
 & & & \multicolumn{3}{c}{Within estimator, coverage (\%)} & \multicolumn{2}{c}{Jackknife, coverage (\%)} & \multicolumn{2}{c}{Median bias}\\
\cmidrule(lr){4-6}\cmidrule(lr){7-8}\cmidrule(lr){9-10}
Design & $h$ & $F$ & AR, quarter & AR, two-way & Wald, two-way & AR, two-way & Wald, two-way & Within & Jackknife\\\midrule
Persistent, $\rho_g=0.67$ & 0 & 35.2 & 92.4 & 92.8 & 92.8 & 93.2 & 93.8 & -0.01 & -0.03\\
 & 4 &  & 93.2 & 94.0 & 91.6 & 94.0 & 92.0 & -0.10 & -0.00\\
 & 7 &  & 93.2 & 93.6 & 92.0 & 93.2 & 92.2 & -0.21 & 0.10\\
 & 12 &  & 92.6 & 93.0 & 91.6 & 94.2 & 92.8 & -0.26 & 0.09\\
\addlinespace
Persistent, $\rho_g=0.9$ & 0 & 10.8 & 88.8 & 89.8 & 93.6 & 91.8 & 95.4 & 0.06 & 0.04\\
 & 4 &  & 92.0 & 92.4 & 93.6 & 93.0 & 94.2 & -0.08 & 0.06\\
 & 7 &  & 91.8 & 92.4 & 92.6 & 93.2 & 94.0 & -0.15 & 0.06\\
 & 12 &  & 91.6 & 92.2 & 91.2 & 93.2 & 93.2 & -0.28 & 0.05\\
\addlinespace
\bottomrule\end{tabular}}
\caption{Coverage of nominal 95\% sets for the response of the outcome at horizon $h$ to a unit change in the policy variable, and median bias of the point estimate. $F$: median first-stage statistic, two-way clustered. AR: Anderson-Rubin set clustered by quarter or two-way. Wald: delta-method interval, two-way clustered. Within: within estimator of the PVAR coefficients. Jackknife: half-panel jackknife correction of the within estimator. Calibrated design with autocorrelation $\rho_g$ of the aggregate shift. 500 replications per design.\label{tab:mc_persistent}}
\end{table}

\section{PVAR-IV and local projections\label{sec:AppLP}}

With an external instrument, the PVAR-IV and local projections instrumented with the same instrument (LP-IV) identify the policy shock in the same way: the first stage of LP-IV is the policy equation of the PVAR, and both use the covariance of the instrument with the policy innovation. They differ in how they estimate the response beyond impact. The PVAR-IV iterates the estimated one-step dynamics of all variables, $\Phi^hb$, while LP-IV projects the outcome at $t+h$ directly on the instrumented policy at $t$, separately for each horizon. This section shows that the two coincide on impact, that beyond impact LP-IV requires an assumption the PVAR-IV does not need, and compares their precision and bias.

To introduce the notation, the LP-IV regression at horizon $h$ for outcome $j$ regresses $Y_{j,i,t+h}$ on $W_{it}$, instrumented by $Z_{it}$, with the $p$ lags of $x_{it}$ and the unit effects as controls. For the cumulated response, the dependent variable is $\sum_{q=0}^{h}Y_{j,i,t+q}$. Let $\kappa_r=\E[Z_{it}\eta_{i,t+r}]$ be the covariance of the instrument with the policy shock $r$ periods later, and $C_s$ the moving average matrices of the PVAR, with $C_0=I_m$.

\begin{prop}[Equivalence on impact]\label{prop:LPimpact}
The LP-IV estimate at $h=0$ equals the PVAR-IV estimate of the impact response, $\beta_j^{IV}$ of Lemma~\ref{lem:FWL}, in every sample.
\end{prop}

\begin{proof}
At $h=0$ the LP-IV regression has the same dependent variable, endogenous regressor, instrument and controls as the instrumental variable regression of Lemma~\ref{lem:FWL}, whose controls are the regressors of the PVAR. The two estimates are therefore the same number, and by the Frisch-Waugh-Lovell theorem they equal $\sum_{i,t}Z_{it}\tilde{Y}_{j,it}/\sum_{i,t}Z_{it}\tilde{W}_{it}$.
\end{proof}

\begin{prop}[Lead exogeneity]\label{prop:LPlead}
Suppose Assumptions~\ref{assu:Linearity} to~\ref{assu:Regularity} hold, with $\kappa_0\neq0$, and abstract from the estimation of the unit effects. Then the LP-IV estimand at horizon $h$ is
\[
\beta^{LP}_{j,h}=\sum_{s=0}^{h}e_j'C_s\,b\,\frac{\kappa_{h-s}}{\kappa_0},
\]
while the PVAR-IV estimand is $e_j'C_hb$. The two coincide at every horizon $h\leq H$ if the instrument is uncorrelated with the policy shocks of the next $H$ periods, $\kappa_r=0$ for $r=1,\dots,H$, and they differ in general otherwise. The same holds for the cumulated responses.
\end{prop}

\begin{proof}
By Assumption~\ref{assu:Shock}(i), the innovations are serially uncorrelated and the PVAR is a VAR($p$) in which $x_{i,t+h}$ equals a linear function of the unit effect and of $x_{i,t-1},\dots,x_{i,t-p}$, plus $\sum_{s=0}^{h}C_s\tilde{x}_{i,t+h-s}$. The controls absorb the first part. The innovations $\tilde{x}_{i,t+h-s}$ for $s=0,\dots,h$ are uncorrelated with the controls, so the residual of $Y_{j,i,t+h}$ on the controls is $\sum_{s=0}^{h}e_j'C_s\tilde{x}_{i,t+h-s}$, and the residual of $W_{it}$ is $\tilde{W}_{it}$. With $\tilde{x}_{it}=b\eta_{it}+\nu_{it}$ and $b_1=1$, Assumption~\ref{assu:Shock}(ii) gives $\E[Z_{it}\tilde{W}_{it}]=\kappa_0$ and $\E[Z_{it}\tilde{x}_{i,t+h-s}]=b\,\kappa_{h-s}$, because the instrument is uncorrelated with $\nu$ at every date. By the Frisch-Waugh-Lovell theorem, the LP-IV estimand is the ratio of $\E[Z_{it}\sum_{s}e_j'C_s\tilde{x}_{i,t+h-s}]=\sum_se_j'C_sb\,\kappa_{h-s}$ and $\kappa_0$, which is the expression in the proposition. The term $s=h$ is $e_j'C_hb$, the PVAR-IV estimand of Lemma~\ref{lem:IRF}, and the other terms vanish when $\kappa_r=0$ for $r=1,\dots,h$. Summing over horizons gives the result for the cumulated responses.
\end{proof}

Proposition~\ref{prop:LPlead} formalizes the observation of \citet{StockWatson2018} that local projections instrumented by an external instrument require lead-lag exogeneity. Under the assumptions of the paper, the lag part is not needed because the controls absorb the past shocks, but the lead part is: an instrument that carries information on future policy shocks, as announcements and other forms of anticipation do, contaminates the local projection at every horizon after impact, with the response to the future shocks it anticipates. The PVAR-IV uses the instrument only on impact, where future shocks do not enter the innovations, and it remains consistent. The converse advantage belongs to LP-IV: Proposition~\ref{prop:LPlead} uses Assumption~\ref{assu:Shock}(i), the recoverability of the shocks from current and past data, which LP-IV does not need when the instrument is uncorrelated with past and future shocks \citep{StockWatson2018}.

Two further differences concern estimation. First, when the PVAR is correctly specified, iterating the estimated dynamics uses the restrictions of the model at every horizon and gives responses with lower variance than the direct projections, at the cost of a bias when the lag order is too short \citep{plagborg2021local,li2024local}. Second, in panels with unit effects, the local projection at horizon $h$ has a dependent variable that leads the regressors by $h$ periods, so removing the unit means creates a bias of the kind described by \citet{nickell1981biases}, even without lagged dependent variables, which grows with the horizon \citep{mei2026nickell}. \citet{mei2026nickell} remove it with a split-panel jackknife, the analogue for local projections of the correction of Section~\ref{sec:Inference}, which the PVAR-IV applies once, to the one-step coefficients.

Table~\ref{tab:mc_lp} compares the two estimators on simulated panels. The calibrated and persistent designs are those of Section~\ref{sec:MonteCarlo}, on the same panels. The anticipated design adds to the shift in the instrument its next innovation, $Z_{it}=s_i(g_t+a\,e_{t+1})$ with $a=1$, so that the instrument is correlated with the next policy shock and lead exogeneity fails, while Assumption~\ref{assu:Shock} holds. The local projection is estimated as in \citet{chodorowreich2021}, without and with the split-panel jackknife correction. On impact the local projection and the PVAR-IV give the same estimates, as Proposition~\ref{prop:LPimpact} states. In the anticipated design, the local projections are strongly biased beyond impact and their sets cover the true response far less often than their nominal level, while the PVAR-IV remains centered on the true response with coverage close to the nominal level, as Proposition~\ref{prop:LPlead} predicts. The jackknife correction does not help in this design, because the bias comes from the correlation of the instrument with future policy shocks and does not vanish with the length of the panel. In the calibrated and persistent designs, where lead exogeneity holds, the local projection underestimates the response at longer horizons, increasingly so with a persistent shift, and its sets under-cover. The jackknife correction reduces this bias, more in the calibrated design than in the persistent one, but it increases the dispersion of the estimates. In these designs the PVAR-IV has a much smaller median absolute error than the local projections at every horizon after impact, and much shorter sets.

\IfFileExists{tables/table_cns_mc_lp.tex}{\begin{table}[ht]\centering\footnotesize\setlength{\tabcolsep}{3pt}
\resizebox{\textwidth}{!}{\begin{tabular}{llr rrr rrr rr rr}\toprule
 & & & \multicolumn{3}{c}{Median bias} & \multicolumn{3}{c}{Median abs.\ error} & \multicolumn{2}{c}{AR coverage (\%)} & \multicolumn{2}{c}{Median AR length}\\
\cmidrule(lr){4-6}\cmidrule(lr){7-9}\cmidrule(lr){10-11}\cmidrule(lr){12-13}
Design & $h$ & $F$ & PVAR & LP & LP-J & PVAR & LP & LP-J & PVAR & LP & PVAR & LP\\\midrule
Calibrated & 0 & 95.1 & 0.01 & 0.01 & 0.04 & 0.38 & 0.38 & 0.40 & 95.4 & 95.4 & 2.25 & 2.24\\
 & 4 &  & 0.08 & -0.23 & 0.01 & 0.50 & 1.02 & 1.02 & 93.8 & 88.6 & 2.77 & 5.23\\
 & 7 &  & 0.09 & -0.45 & -0.24 & 0.72 & 1.46 & 1.62 & 92.0 & 87.4 & 3.67 & 7.04\\
 & 12 &  & 0.05 & -0.83 & -0.40 & 0.67 & 1.81 & 2.32 & 92.2 & 86.8 & 3.62 & 8.54\\
\addlinespace
Persistent, $\rho_g=0.9$ & 0 & 17.3 & 0.04 & 0.06 & 0.02 & 0.74 & 0.77 & 0.93 & 91.8 & 89.8 & 5.48 & 4.87\\
 & 4 &  & 0.06 & -0.43 & -0.36 & 0.68 & 3.31 & 3.95 & 93.0 & 77.6 & 4.89 & 14.38\\
 & 7 &  & 0.06 & -1.84 & -1.05 & 0.93 & 5.16 & 6.06 & 93.2 & 68.8 & 6.28 & 20.69\\
 & 12 &  & 0.05 & -3.25 & -2.17 & 0.86 & 6.92 & 9.10 & 93.2 & 64.6 & 6.16 & 27.88\\
\addlinespace
Anticipated, $a=1$ & 0 & 29.8 & -0.00 & 0.00 & 0.07 & 0.54 & 0.53 & 0.59 & 93.2 & 93.0 & 3.43 & 3.40\\
 & 4 &  & -0.04 & 2.43 & 2.56 & 0.55 & 2.61 & 2.71 & 93.6 & 69.8 & 3.42 & 8.88\\
 & 7 &  & -0.03 & 2.20 & 2.60 & 0.72 & 2.84 & 3.11 & 94.4 & 75.8 & 4.40 & 11.76\\
 & 12 &  & -0.03 & 0.79 & 1.41 & 0.77 & 3.23 & 4.20 & 93.2 & 82.2 & 4.36 & 13.88\\
\addlinespace
\bottomrule\end{tabular}}
\caption{PVAR-IV and local projection IV of the cumulated response at horizon $h$, on the same simulated panels. PVAR: within estimator with the half-panel jackknife correction. LP: change of the level of the outcome from $t-1$ to $t+h$ regressed on the policy instrumented by $Z_{it}$, with $p$ lags of both variables as controls and county and quarter effects. LP-J: the same with the split-panel jackknife correction. Anticipated: the instrument also carries the next innovation of the aggregate shift, with loading $a$. Sets: nominal 95\% Anderson-Rubin sets with two-way clustering. Median absolute error: median of the absolute deviation of the estimate from the true response. $F$: median first-stage statistic of the local projection, two-way clustered. 500 replications per design.\label{tab:mc_lp}}
\end{table}
}{}

\section{Spillovers and shift-share instruments\label{sec:AppSpill}}

Assumption~\ref{assu:Linearity} lets the outcomes of a unit depend only on its own policy. In regional data this rules out interference: stock wealth gains in one county can raise spending, and therefore payroll, in neighbouring counties. A large literature studies causal effects under interference \citep{hudgens2008toward,aronow2017estimating}, and models the dependence of the outcome of a unit on the treatments of other units through a known exposure mapping \citep{forastiere2021identification,forastiere2022estimating}, also in encouragement designs, where an instrument moves the treatment \citep{forastiere2016identification,forastiere2021exploring,hoshino2024causal}. This section shows that interference makes shift-share instruments fail, and that a known exposure mapping restores identification.

To introduce the notation, let $r_{ik}\geq0$ be known weights with $r_{ii}=0$ and $\sum_kr_{ik}=1$, such as indicators of neighbouring counties divided by their number, and let $S_{it}=\sum_kr_{ik}W_{kt}$ be the exposure of unit $i$ to the policies of the other units, and $Z^S_{it}=\sum_kr_{ik}Z_{kt}$ the corresponding exposure to their instruments. I consider a shift-share instrument $Z_{it}=s_ig_t$ in a model with period effects, as in Section~\ref{sec:Theorems}, and write $\bar s$ for the average exposure, $\tilde g_t$ for the shift net of its conditional mean, and $\sigma_g^2$ for its variance.

\begin{assumption}[Exposure mapping]\label{assu:Exposure}
For each $(i,t)$ there are random coefficients $\theta_{j,it}$, $\varphi_{j,it}$ and $\pi_{it}$ such that the potential outcomes depend on the policies of all units through
\[
Y_{j,it}(w_{1t},\dots,w_{Nt})=Y_{j,it}(0,0)+\theta_{j,it}\,w_{it}+\varphi_{j,it}\sum_kr_{ik}w_{kt},\qquad W_{it}(z)=W_{it}(0)+\pi_{it}\,z,
\]
and the observed variables are $W_{it}=W_{it}(Z_{it})$ and $Y_{j,it}=Y_{j,it}(W_{1t},\dots,W_{Nt})$.
\end{assumption}

\begin{assumption}[Independence under interference]\label{assu:IndepSpill}
The potential outcomes $Y_{j,it}(0,0)$, the policies $W_{kt}(0)$ and the coefficients $\theta_{j,it}$, $\varphi_{j,it}$ and $\pi_{kt}$ of all units are independent of the shift $g_t$ given the information sets, and $\var(\tilde g_t)=\sigma_g^2>0$.
\end{assumption}

Assumption~\ref{assu:Exposure} keeps the first stage free of interference, since the policy of a unit responds only to its own instrument, and lets the outcome respond to the average policy of the units connected to it, with the spillover coefficient $\varphi_{j,it}$. Assumption~\ref{assu:IndepSpill} is Assumption~\ref{assu:Independence} stated for the whole cross-section, since the shift enters the instruments of all units.

\begin{prop}[Spillover bias]\label{prop:SpillBias}
Under Assumptions~\ref{assu:Exposure} and~\ref{assu:IndepSpill}, with $Z_{it}=s_ig_t$ and period effects,
\[
\frac{\Gamma_j}{\Gamma_1}=\frac{\sum_i(s_i-\bar s)^2\,\E[\theta_{j,it}\pi_{it}\mid i]}{\sum_i(s_i-\bar s)^2\,\E[\pi_{it}\mid i]}+\frac{\sum_i\sum_kr_{ik}(s_i-\bar s)(s_k-\bar s)\,\E[\varphi_{j,it}\pi_{kt}\mid i,k]}{\sum_i(s_i-\bar s)^2\,\E[\pi_{it}\mid i]}.
\]
\end{prop}

\begin{proof}
By Assumption~\ref{assu:Exposure}, the outcome is $Y_{j,it}=Y_{j,it}(0,0)+\theta_{j,it}W_{it}+\varphi_{j,it}\sum_kr_{ik}W_{kt}$, and each policy is $W_{kt}=W_{kt}(0)+\pi_{kt}Z_{kt}$. Since the period effects remove $\bar s g_t$, the part of $Z_{kt}$ that is not predictable from the information sets is $(s_k-\bar s)\tilde g_t$. Taking residuals, multiplying by $(s_i-\bar s)\tilde g_t$ and taking expectations, the terms in $Y_{j,it}(0,0)$ and $W_{kt}(0)$ vanish by Assumption~\ref{assu:IndepSpill}, and so do the products of the coefficients with the predictable part of the instruments. What remains is $\Gamma_1=\sigma_g^2N^{-1}\sum_i(s_i-\bar s)^2\E[\pi_{it}\mid i]$ and
\[
\Gamma_j=\sigma_g^2N^{-1}\Bigl(\sum_i(s_i-\bar s)^2\E[\theta_{j,it}\pi_{it}\mid i]+\sum_i\sum_kr_{ik}(s_i-\bar s)(s_k-\bar s)\E[\varphi_{j,it}\pi_{kt}\mid i,k]\Bigr),
\]
where the second sum comes from $\E[\tilde g_t^2]=\sigma_g^2$ multiplying the exposures of unit $i$ and of the units connected to it. Their ratio is the expression in the proposition.
\end{proof}

The first term is the $\mu$-LATE of Section~\ref{sec:Theorems}. The second is a spillover bias, which vanishes when there are no spillovers, $\varphi_{j,it}=0$, or when the exposure of a unit is uncorrelated with the exposures of the units connected to it. Neither condition is plausible for regional shift-share designs, where exposures are geographically clustered: wealthy counties neighbour wealthy counties, so a positive spillover makes the instrumental variable estimand overstate the direct effect. The source of the bias is the aggregate shift. The shift that moves the instrument of a unit moves at the same time the instruments, and so the policies, of all the units connected to it, which violates the exclusion restriction even though the policy of each unit responds only to its own instrument. With instruments that are independent across units, the cross products $(s_i-\bar s)(s_k-\bar s)\sigma_g^2$ would be replaced by zero covariances, and the bias would disappear.

A known exposure mapping restores identification, because the exposure $Z^S_{it}$ of a unit to the instruments of the other units is itself an instrument for its exposure $S_{it}$ to their policies.

\begin{prop}[Identification with a known exposure mapping]\label{prop:SpillId}
Suppose Assumptions~\ref{assu:Exposure} and~\ref{assu:IndepSpill} hold with constant effects, $\theta_{j,it}=\theta_j$ and $\varphi_{j,it}=\varphi_j$, while $\pi_{it}$ may vary, and let $\tilde W_{it}$, $\tilde S_{it}$ and $\tilde Y_{j,it}$ be the residuals of the PVAR augmented with $S_{it}$. If the matrix
\[
Q=\E\begin{bmatrix}Z_{it}\tilde W_{it} & Z_{it}\tilde S_{it}\\ Z^S_{it}\tilde W_{it} & Z^S_{it}\tilde S_{it}\end{bmatrix}
\]
is nonsingular, the two-stage least squares regression of $\tilde Y_{j,it}$ on $(\tilde W_{it},\tilde S_{it})$ with instruments $(Z_{it},Z^S_{it})$ converges to $(\theta_j,\varphi_j)$. With a constant first stage, $\pi_{it}=\pi$, $Q$ is proportional to the second-moment matrix of $\bigl(s_i-\bar s,\ \sum_kr_{ik}(s_k-\bar s)\bigr)$ across units, and it is nonsingular unless the exposure of each unit is proportional to the average exposure of the units connected to it.
\end{prop}

\begin{proof}
With constant effects, Assumption~\ref{assu:Exposure} gives $Y_{j,it}=Y_{j,it}(0,0)+\theta_jW_{it}+\varphi_jS_{it}$. The residuals on the information set are linear in the variables, so $\tilde Y_{j,it}=\tilde Y_{j,it}(0,0)+\theta_j\tilde W_{it}+\varphi_j\tilde S_{it}$, where $\tilde Y_{j,it}(0,0)$ is the residual of the untreated outcome. By Assumption~\ref{assu:IndepSpill}, $\E[Z_{it}\tilde Y_{j,it}(0,0)]=\E[Z^S_{it}\tilde Y_{j,it}(0,0)]=0$, so the two moment conditions $\E[(Z_{it},Z^S_{it})'(\tilde Y_{j,it}-\theta_j\tilde W_{it}-\varphi_j\tilde S_{it})]=0$ hold at the true values. They have the unique solution $(\theta_j,\varphi_j)'=Q^{-1}\E[(Z_{it},Z^S_{it})'\tilde Y_{j,it}]$ when $Q$ is nonsingular, and the two-stage least squares estimator is its sample analogue, which converges by the law of large numbers. With $\pi_{it}=\pi$, the computation in the proof of Proposition~\ref{prop:SpillBias} gives each entry of $Q$ as $\pi\sigma_g^2N^{-1}$ times a sum over units of products of $s_i-\bar s$ and $\sum_kr_{ik}(s_k-\bar s)$, which is the second-moment matrix in the proposition.
\end{proof}

In the PVAR, the exposure enters as an additional variable, $x_{it}=(W_{it},S_{it},Y_{1,it},\dots,Y_{J,it})'$, computed from the policies of the other units with the known weights. The policy shock of a unit does not move the policies of the other units on impact, so the impact vector of the policy shock is $(1,0,\theta_1,\dots,\theta_J)'$, and the responses at longer horizons follow from the estimated dynamics as in Lemma~\ref{lem:IRF}, with the exposure held to the path implied by the system. With heterogeneous effects, two-stage least squares with two endogenous variables recovers combinations of $\theta_{j,it}$ and $\varphi_{j,it}$ whose weights need not be positive, so the constant-effects condition of Proposition~\ref{prop:SpillId} is what gives the estimate a causal interpretation in this case.

\end{document}